%% file: paper.tex
\documentclass[sigplan,10pt,nonacm]{acmart}

\newif\iffull
\fulltrue

\input{preamble}

\setcopyright{acmlicensed}
\acmPrice{15.00}
\acmDOI{10.1145/3385412.3385975}
\acmYear{2020}
\copyrightyear{2020}
\acmSubmissionID{pldi20main-p93-p}
\acmISBN{978-1-4503-7613-6/20/06}
\acmConference[PLDI '20]{Proceedings of the 41st ACM SIGPLAN International Conference on Programming Language Design and Implementation}{June 15--20, 2020}{London, UK}
\acmBooktitle{Proceedings of the 41st ACM SIGPLAN International Conference on Programming Language Design and Implementation (PLDI '20), June 15--20, 2020, London, UK}

\begin{document}

\title{Proving Data-Poisoning Robustness in Decision Trees}         
\author{Samuel Drews}
\affiliation{
    \institution{University of Wisconsin-Madison}
    \city{Madison}
    \state{WI}
    \country{USA}
}
\email{sedrews@wisc.edu}
\author{Aws Albarghouthi}
\affiliation{
    \institution{University of Wisconsin-Madison}
    \city{Madison}
    \state{WI}
    \country{USA}
}
\email{aws@cs.wisc.edu}
\author{Loris D'Antoni}
\affiliation{
    \institution{University of Wisconsin-Madison}
    \city{Madison}
    \state{WI}
    \country{USA}
}
\email{loris@cs.wisc.edu}

\begin{abstract}
Machine learning models are brittle, and
small changes in the training data can result in different predictions.
We study the problem of proving that a prediction is robust
to \emph{data poisoning}, where an attacker can 
inject a number of malicious elements into the training set to influence the learned model.
We target decision-tree models, a popular and simple class of machine 
learning models that underlies many complex learning techniques.
We present a sound verification technique based on \emph{abstract interpretation}
and implement it in a tool called \name.
\name abstractly trains decision trees for an intractably 
large space of possible poisoned datasets.
Due to the soundness of our abstraction, \name can produce proofs that, for a given input, 
the corresponding prediction would not have 
changed had the training set been tampered with or not.
We demonstrate the effectiveness of \name on a number
of popular datasets.
\end{abstract}



\maketitle

\input{introduction}
\input{overview}

\input{abslearning}
\input{absdomain}

\input{extensions}
\input{evaluation}

\input{relatedwork}
\input{conclusion}

\bibliography{paper}

\iffull
\appendix

\input{app-proofs}
\input{app-real}
\input{data/bench}
\fi

\end{document}

%% file: preamble.tex
\usepackage{amsmath}
\usepackage{amsfonts}
\usepackage{dsfont}
\usepackage{mathtools}
\usepackage{xspace}
\usepackage{xcolor}
\usepackage{etoolbox}
\usepackage[linesnumbered,noend]{algorithm2e}
\usepackage{pgfplots}
\usepgfplotslibrary{groupplots}
\usepackage{pgffor}
\usepackage{listings}
\usepackage{stmaryrd}
\usepackage{wrapfig}
\usepackage{latexsym}
\usepackage{multirow}
\usepackage[autolanguage]{numprint}
\usepackage{enumitem}
\usepackage{microtype}

\usetikzlibrary{arrows}

\usepackage{booktabs}   
\usepackage{subcaption} 

\usepackage[T1]{fontenc}

\renewcommand{\paragraph}[1]{\vspace{.06in}\noindent\textbf{{#1}.}}

\newcommand{\name}{Antidote\xspace}

\newcommand{\rone}{(\emph{i})~}
\newcommand{\rtwo}{(\emph{ii})~}
\newcommand{\rthree}{(\emph{iii})~}
\newcommand{\rfour}{(\emph{iv})~}

\renewcommand{\leq}{\leqslant}

\newcommand{\restrsym}{\downarrow}
\newcommand{\restr}[2]{{#1{\restrsym_{#2}}}}
\newcommand{\arestr}[2]{{#1{\restrsym_{#2}^{\#}}}}

\DeclareMathOperator*{\argmax}{argmax}
\DeclareMathOperator*{\argmin}{argmin}
\DeclareMathOperator*{\lb}{lb}
\DeclareMathOperator*{\ub}{ub}

\newcommand{\train}{T}
\newcommand{\tdom}{\mathcal{X}}
\newcommand{\point}{x}
\newcommand{\class}{y}
\newcommand{\classes}{\mathcal{Y}}
\newcommand{\drop}[2]{\langle#1,#2\rangle}

\newcommand{\reals}{\mathbb{R}}
\newcommand{\learner}{L}
\newcommand{\tlearner}{\textsf{DTrace}}
\newcommand{\atlearner}{\tlearner^\#}
\newcommand{\tlearnerr}{\tlearner_\reals}
\newcommand{\atlearnerr}{\atlearner_\reals}

\newcommand{\model}{M}
\newcommand{\poison}{\Delta}
\newcommand{\Drop}{\Delta}

\newcommand{\impurity}{\mathsf{ent}}
\newcommand{\aimpurity}{\mathsf{ent}^\#}

\newcommand{\summary}{\mathsf{cprob}}
\newcommand{\asummary}{\mathsf{cprob}^\#}

\newcommand{\bestsplit}{\mathsf{bestSplit}}
\newcommand{\abestsplit}{\mathsf{bestSplit}^\#}
\newcommand{\bestsplitr}{\bestsplit_\reals}
\newcommand{\abestsplitr}{\abestsplit_\reals}

\newcommand{\lub}{lub}

\newcommand{\filter}{\mathsf{filter}}
\newcommand{\afilter}{\mathsf{filter}^\#}
\newcommand{\filterr}{\filter_\reals}
\newcommand{\afilterr}{\afilter_\reals}

\newcommand{\score}{\mathsf{score}}
\newcommand{\ascore}{\mathsf{score}^\#}
\newcommand{\phibest}{\varphi^\star}
\newcommand{\rpred}{\rho}

\newcommand{\tree}{R}

\newcommand{\seq}{\sigma}

\newcommand{\ifs}{\textbf{if}}
\newcommand{\thens}{\textbf{then}}
\newcommand{\elses}{\textbf{else}}

\newcommand{\nullpred}{\diamond}

\newcommand{\numBenchmarks}{5}
\newcommand{\iris}{Iris\xspace}
\newcommand{\wdbc}{Wisconsin Diagnostic Breast Cancer\xspace}
\newcommand{\mammography}{Mammographic Masses\xspace}
\newcommand{\mnistbin}{MNIST-1-7-Binary\xspace}
\newcommand{\mnistreal}{MNIST-1-7-Real\xspace}

\newcommand{\puresets}{\mathit{pure}}

\renewcommand{\qed}{}

%% file: introduction.tex
\section{Introduction}\label{sec:introduction}

Artificial intelligence, in the form of machine learning (ML), is rapidly transforming the world as we know it. Today, ML is responsible for an ever-growing spectrum of sensitive decisions---from loan decisions, to diagnosing diseases, to autonomous driving. 
Many recent works have shown how ML models are brittle~\cite{SzegedyZSBEGF13,DBLP:conf/icml/WangJC18,chen2017targeted,Biggio12,steinhardt2017certified}, and
with ML spreading across many industries, the issue of robustness in ML models has taken center stage.
The research field that deals with studying robustness of ML models is referred to as
\emph{adversarial machine learning}.
In this field, researchers have proposed many definitions that try to capture robustness to different \emph{adversaries}.
The majority of these works have focused on verifying or improving the model's robustness to \emph{test-time attacks}~\cite{Gehr18,singh2019abstract,anderson2019optimization,katz2017reluplex,wang2018formal},
where
an adversary can craft small perturbations to input examples that
fool the ML model into changing its prediction,
e.g., making the model think a picture of a cat is that of a zebra~\cite{carlini2017towards}.

\paragraph{Data-Poisoning Robustness}
This paper focuses on verifying \emph{data-poisoning robustness},
which captures how robust a training algorithm $\learner$ is to variations in a given \emph{training set}
$\train$. 
Intuitively, applying $\learner$ to the training set $\train$ results in a classifier (model)
$\model$, and in this paper we are interested in how the trained model varies when
producing perturbations of the input training set $\train$. 

The idea is that an adversary can produce slight modifications
of the training set, e.g., by supplying a small amount of malicious training points,
to influence the produced model and its predictions.
This attack model is possible when data is curated, for example,
via crowdsourcing or from online repositories, where attackers can try to add
malicious elements to the training data.
For instance,~\citet{xiao2015feature} consider adding malicious training points to
affect a malware detection model; similarly,~\citet{chen2017targeted} 
consider adding a small number of images to bypass a facial recognition model.

\begin{figure}[t]
  \includegraphics[width=8.5cm]{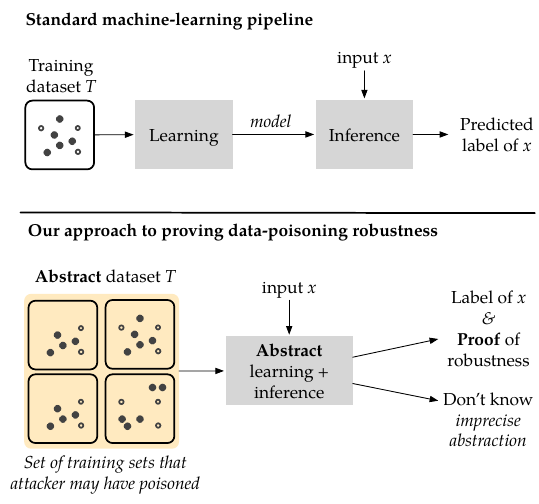}
  \caption{High-level overview of our approach}\label{fig:general}
\end{figure}

A \emph{perturbed set} $\poison(\train)$
defines a set of neighboring datasets the adversary could have attacked to yield the training set $\train$.
To define what it means for the training algorithm $\learner$ to be robust, we need to measure
how the model learned by $\learner$ varies when modifying the training set.
Let us say  we have an input example $\point$---e.g., a test example---and its classification label is $M(\point)=y$,
where $M=\learner(\train)$ is the model learned from the  training set $\train$.
We say that $\point$ is robust to poisoning if and only if
for all $\train' \in \poison(T)$, we have $\learner(\train')(\point) = \class$;
that is, no matter what dataset $\train' \in \poison(\train)$ we use to construct the model $\model' = \learner(\train')$,
we want $\model'$ to always return the same classification $\class$ on the input $\point$.
To be clear, in this paper we are concerned with a \emph{local} robustness property:
we are proving the invariance of individual test points' classifications
to changes in the training set.


\paragraph{Verification Challenges}
Data-poisoning robustness has been studied extensively~\cite{Biggio12,Xiao12,Xiao15,Newell14,Mei15}. This body
of work has demonstrated data-poisoning attacks---i.e., modifications to training sets---that can degrade classifier accuracy, sometimes
dramatically, or force certain predictions
on specific inputs.
 While some defenses have been proposed against specific attacks~\cite{LaishramP16,steinhardt2017certified}, 
we are not aware of any technique that can formally verify that a given learning algorithm is robust to perturbations
to a given training set.
Verifying data-poisoning robustness of a given learner requires solving a number of challenges:
\begin{enumerate}[topsep=5pt, partopsep=0pt, leftmargin=*]
\item The datasets over which the learner operates are typically large (thousands of elements).
Even when considering simple poisoning attacks, the number of modified training sets we need to consider
        can be intractably large to represent and explore explicitly.
\item Because learners are complicated programs that employ complex metrics (e.g., entropy  and loss functions),
their verification requires new specialized techniques.			
\end{enumerate}

\paragraph{Our Approach}
We focus on the problem of verifying data-poisoning robustness for
\textit{decision-tree learners}. We choose decision trees because
\rone they are widely used interpretable models;
\rtwo they are used in industrial models like random forests and XGBoost~\cite{chen2016xgboost};
\rthree decision-tree-learning  has been shown to be 
unstable to training-set perturbation~\cite{dwyer07,turney95,li02,perez05};
and \rfour decision-tree-learning algorithms are typically deterministic---e.g., they do not employ stochastic optimization techniques---%
making them amenable to verification.

We present \name, \textit{a tool for verifying data-poisoning robustness of decision-tree learners}.
At a high level, \name takes as input a training set $\train$ and an input $\point$,
 symbolically constructs every tree built by a particular decision-tree learner $\learner$ on every possible variation
 of $\train$ in $\poison(\train)$, and applies all those trees to $\point$.
 If all the trees agree on the label of $\point$, then we know that $\point$ is robust to poisoning $\train$.
 (See Figure~\ref{fig:general} for an overview.)
\name addresses the two challenges highlighted above as follows.

To address the first challenge of training on a combinatorially large number of datasets, 
\name employs a novel \emph{abstract domain} for concisely representing sets of datasets.
\name is a sound abstract interpretation
of standard decision-tree learning algorithms:
Instead of constructing a single decision tree,
it implicitly constructs an overapproximation of all decision trees for every
training set in $\poison(\train)$.

To address the second challenge, \name has to soundly approximate  how decision-tree learning algorithms propagate entropy computations across a model.
\name takes advantage of the following observation: 
every input $\point$ only \emph{traverses} a single root-to-leaf trace in the learned decision tree---i.e., the sequence of predicates that affects the decision for $\point$. 
This observation allows \name to build a simpler abstraction that only needs to track the  predicates and training elements affecting the
decision for $\point$ at a given node in the tree, instead of across the entire tree.
We call this a \emph{trace-based view}
of decision-tree learning, where we are only concerned with the tree trace(s) traversed by $\point$. 

\paragraph{Evaluation}
We evaluated \name  on 
a number of real datasets from the literature.
\name successfully proves robustness for many test inputs, even in cases where the
learner is allowed to build a complex decision tree and
the attacker is allowed to contribute more than 1\% of the points in the training set.
For instance, \name can, in around 1 minute, prove robustness for some test inputs of the MNIST-1-7~\cite{steinhardt2017certified,Biggio12} datasets for cases where the attacker may have contributed up to 192 malicious elements to the dataset.
 A na\"ive enumeration approach would have to construct around $10^{432}$ models to prove the same property!

\paragraph{Contributions}
We summarize our contributions as follows:

\begin{itemize}[topsep=5pt, partopsep=0pt, leftmargin=*]
\item \name: the first sound technique for verifying
    data-poisoning robustness for decision-tree learners (\S\ref{sec:overview}).
\item A \emph{trace-based view} of decision-tree learning
    as a stand-alone algorithm that allows us to sidestep
    the challenging problem of reasoning about the set of all possible output trees
    a learner can output on different datasets (\S\ref{sec:abslearning}).
\item An \textit{abstract domain that concisely encodes sets of perturbed datasets}
    and the abstract transformers necessary to verify
    robustness of a decision-tree learner (\S\ref{sec:absdomain} and \S\ref{sec:extensions}).
\item An evaluation of \name on five representative datasets from the literature.
    \name can prove poisoning robustness for all datasets
    in cases where an enumeration approach would be doomed to fail (\S\ref{sec:evaluation}).
\end{itemize}

Proofs of theorems and figures for additional benchmarks are available in
\iffull Appendix~\ref{app:proofs}. \else
the full version of this paper~\cite{antidotearxiv}. \fi

%% file: overview.tex
\section{Overview}\label{sec:overview}

In this section, we give an overview of decision-tree learning,
the poisoning-robustness problem, and motivate our abstraction-based proof technique.

\paragraph{Decision-Tree Learning}
Consider the dataset $T_{\emph{bw}}$ at the top of Figure~\ref{fig:example}.
It is comprised of 13 elements with a single numerical feature.
Each element is labeled as a white (empty) or  black (solid) circle.
We use $x$ to denote the feature value of each element.
Our goal is to construct a  decision tree that classifies a given number into
white or black.

For simplicity, we assume that we can only build trees of depth 1, like the one shown 
at the bottom Figure~\ref{fig:example}.
At each step of building a decision tree, the learning algorithm is looking for a predicate
$\varphi$ with the best score, with the goal of splitting the dataset into two pieces with \emph{least diversity}, i.e., most elements have the same class (formally defined usually using a notion of entropy).
This is what we see in our example: using the predicate $x \leq 10$,
we split the dataset into two sets, one that is mostly white (left) and one that is completely black (right).
This is the best split we can have for our data, assuming we can only pick predicates of the form
$x \leq c$, for an integer $c$.\footnote{
Note that, while the set of predicates $x \leq c$ is infinite, 
for this dataset (and in general for any dataset), there exists only finitely
many inequivalent predicates---e.g., $x\leq 4$ and $x\leq 5$ split the 
dataset into the same two sets.
}

Given a new element for a classification, we check if it is $\leq 10$,
in which case we say it is white with probability $7/9$---i.e.,
the fraction of white elements such that $\leq 10$.
Otherwise, if the  element is $>10$, we say it is black with probability 1.

\begin{figure}[t]
    \includegraphics[width=8cm]{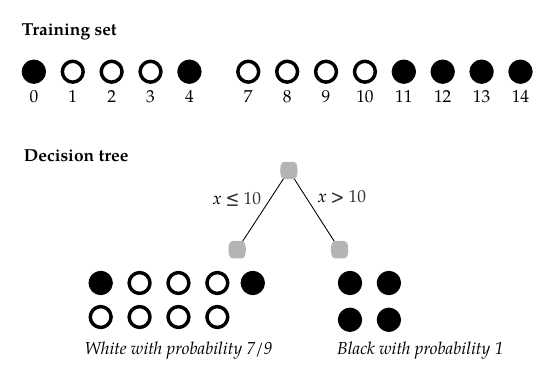}
    \caption{Illustrative example}\label{fig:example}
\end{figure}

\paragraph{Data-Poisoning Robustness}
Imagine we want to classify an input $x$
but want to make sure the classification would not have changed had the training data
been slightly different.
For example, maybe some percentage of the data was maliciously added by an attacker
to sway the learning algorithm, a problem known as \emph{data poisoning}.
Our goal is to check whether the classification of $x$ is robust to data poisoning.

\paragraph{A Na\"ive Approach}
Consider our running example and imagine  we want to classify the number $5$.
Additionally, we want to prove that \emph{removing up to two elements} from the training set
would not change the classification of $5$---i.e., we assume that up to ${\sim}15\%$ (or $2/13$) of the dataset is contributed maliciously.
The na\"ive way to do this is to consider every possible training dataset
with up to two elements removed and retrain the decision tree.
If all trees classify the input $5$ as white, the classification is robust
to this level of poisoning.

Unfortunately, this approach is intractable. Even for our tiny example, 
we have to train 92 trees $(\binom{13}{2} + \binom{13}{1} + 1)$.
For a dataset of 1000 elements and a poisoning of up to 10 elements,
we have ${\sim}10^{23}$ possibilities.

\paragraph{An Abstract Approach}
Our approach to efficiently proving poisoning robustness exploits a number
of insights.
First, we can perform decision-tree learning \emph{abstractly} on a \emph{symbolic set of training sets},
without having to deal with a combinatorial explosion.
The idea is that the operations in decision-tree learning, e.g., selecting a predicate and splitting the dataset, do not 
need to look at every concrete element of a dataset, but at aggregate statistics (counts).

Recall our running example in Figure~\ref{fig:example}. Let us say that up to two elements have been removed.
No matter what two elements you choose, the predicate $x \leq 10$ remains one that gives \emph{a} best split for the dataset.
In cases of ties between predicates,
our algorithm abstractly represents all possible splits.
%
For each predicate, we can symbolically compute best- and worst-case scores 
in the presence of poisoning as an \emph{interval}.
Similarly, we can also compute an interval that overapproximates
the set of possible classification probabilities.
For instance, in the left branch of the decision-tree, the probability will be
$[0.71,1]$ instead of $0.78$ (or $7/9$).
The best case probability of 1 is when we drop the black points $0$ and $4$;
the worst-case probability of $0.71$ (or $5/7$) is when we drop any two white points.

The next insight that enables our approach is that we \emph{do not need to explicitly build the tree}.
Since our goal is to prove robustness of a single input point, which effectively takes a single trace through the tree, we mainly need to keep track of the abstract training sets as they propagate along those
traces. This insight drastically simplifies our approach; otherwise, we would need to somehow abstractly represent sets of elements of a tree data structure, a non-trivial problem in program analysis.

\begin{figure}[t]
    \includegraphics[width=1in]{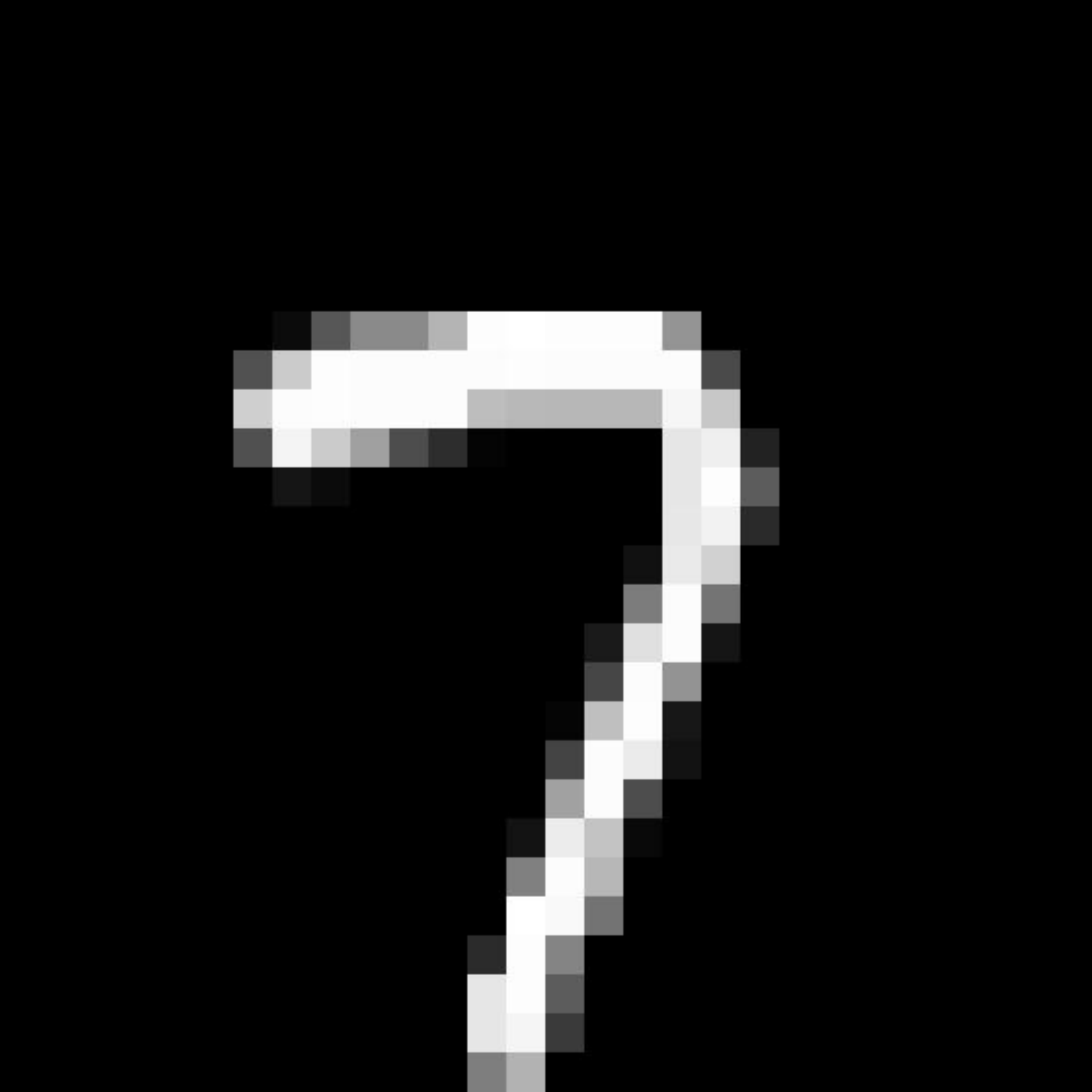}
    \caption{Example MNIST-1-7 digit that is proven poisoning-robust by \name.}
    \label{fig:exmnist}
\end{figure}

\paragraph{Abstraction and Imprecision}
We note that our approach is sound but necessarily incomplete;
that is, 
when our approach returns ``robust'' the answer is correct, but
there are robust instances for which our approach will not be able to prove robustness.
The are numerous sources of imprecision due to overapproximation, for example,
we use the \emph{intervals domain} (or disjunctive intervals) to capture  real-valued entropy calculations of different training set splits, as well as the final probability of classification.

\paragraph{An Involved Example} 
To further illustrate our technique,
we preview one of our experiments.
We applied our approach to the MNIST-1-7 dataset,
which has been 
used to study data-poisoning for deep neural networks~\cite{steinhardt2017certified}
and support vector machines~\cite{Biggio12}.
In our experiments, we checked whether \name could prove data-poisoning for the inputs used in the same
dataset when training a decision tree.
For example, when applying \name to the image of the digit in Figure~\ref{fig:exmnist}, \name proves that it is poisoning robust
(always classified as a seven)
for up to 192 poisoned elements in 90 seconds.
This is equivalent to training on ${\sim}10^{432}$ datasets!

%% file: abslearning.tex
\section{Poisoning and Decision Tree Learning}\label{sec:abslearning}

In this section, we begin by formally defining the \emph{data-poisoning-robustness problem}.
Then, we present a \emph{trace-based} view of decision-tree learning,
which will pave the way for a poisoning-robustness proof technique.

\subsection{The Poisoning Robustness Problem}
In a typical supervised learning setting,
we are given a learning algorithm $\learner$
and a training set $\train \subseteq \tdom \times \classes$
comprised of elements of some set $\tdom$, each with its classification label from
a finite set of classes $\classes$.
Applying $\learner$ to $\train$ results in a classifier (or model):
$\model : \tdom \to \classes$.
For now, we assume  that both the learning algorithm $\learner$
and the models it learns are deterministic functions.%
\footnote{Our approach, however, needs to handle non-determinism in decision-tree learning,
which arises when breaking ties for
choosing predicates with equal scores and choosing labels for classes with equal probabilities.}

A \emph{perturbed set} $\poison(\train) \subseteq 2^{\tdom \times \classes}$
 defines a set of possible \emph{neighboring} datasets of $\train$.
Our robustness definitions are relative to some given perturbation $\poison$.
(In Section~\ref{sec:npoison-model}, we define a specific perturbed set
that captures a particular form of data poisoning.)
%

\begin{definition}[Poisoning Robustness]
Fix a learning algorithm $\learner$, a training set $\train$, and let $\poison(\train)$ be a perturbed set.
Given an element $\point \in \tdom$, we say that $\point$ is robust
to poisoning $\train$ if and only if 
\[
\forall \train' \in \poison(\train) \ldotp \ \learner(\train')(x) = \learner(\train)(x) 
\]
When $\train$ and $\poison$ are clear from context,
we will simply say that $\point$ is robust.
\end{definition}
In other words, no matter what  dataset $\train' \in \poison(\train)$ we use to construct the model $\model = \learner(\train')$,
we want  $\model$ to always return the same classification for $\point$.
{returned for $\point$ by the model $\learner(\train)$
learned on the original training set $\train$.}

\begin{example}
Imagine we suspect that an attacker has contributed 10 training points
to $\train$, but we do not know which ones.
We can define $\poison(\train)$ to be $\train$ as well as every subset of $\train$
of size $|\train|-10$.
If an input $x$ is robust for this definition of $\poison(\train)$,
then no matter whether the attacker has contributed 10 training items or not,
the classification of $x$ does not change.
\end{example}


\subsection{Decision Trees: A Trace-Based View}
We now formally define decision trees.
We will formalize a tree as the \emph{set of traces} from the root
to each of the leaves.
As we will see, this trace-based view will help enable our proof technique.
The idea of representing
an already-learned decision tree as a set of traces
is not new and has often been explored in the context of
extracting interpretable rules from decision trees~\cite{Quinlan87}.

A decision tree $\tree$ is a finite set of traces,
where each trace is a tuple $(\seq, \class)$ such that
$\seq$ is a sequence of Boolean predicates and
 $\class \in \classes$ is the classification.
 
Semantically, a tree $\tree$ is a function in $\tdom \rightarrow \classes$.
Given an input $\point \in \tdom$,
applying $\tree(\point)$ results in a classification $\class$
from the trace $(\seq, \class) \in \tree$ where $x$ satisfies all the predicates
in the sequence $\seq = [\varphi_1,\ldots,\varphi_n]$, that is,
 $\bigwedge_{i=1}^n \point \models \varphi_i$ is true.  
We say a tree $\tree$ is \textit{well-formed} if for every $\point\in\tdom$ there exists exactly 
one
trace $(\seq, \class) \in \tree$ such that  $x$ satisfies all predicates in $\seq$.
In the following we assume all trees are well-formed.

\begin{example}[Decision tree traces]
Consider the decision-tree in Figure~\ref{fig:example}.
It contains two traces, each with a sequence of predicates containing a single predicate:
$([x \leq 10], \text{\it white})$ and
$([x > 10], \text{\it black})$.
\end{example}

\subsection{Decision-Tree Learning: A Trace-Based View}

We now present a simple decision-tree learning algorithm, $\tlearner$.
Then, in Section~\ref{sec:absdomain}, we abstractly interpret $\tlearner$ with the goal of  proving poisoning robustness.

One of our key insights is that we do not need to explicitly represent the learned trees (i.e., the set of all traces),
since our goal is to prove robustness of a \emph{single  input} point, which effectively takes a \emph{single trace} through the tree.
Therefore, in this section, we will define a \emph{trace-based decision-tree learning algorithm}.
This is inspired by standard algorithms---%
like CART~\cite{breiman2017classification}, ID3~\cite{quinlan1986induction}, and C4.5~\cite{quinlan1993c}---%
but \emph{it is input-directed, in the sense that it only builds the trace of the tree that a given
input $\point$ will actually traverse.}

\paragraph{A Trace-Based Learner}
Our trace-based learner $\tlearner$ is shown in Figure~\ref{fig:alg}.
It takes a training set $\train$ and an input $x$
and computes the trace traversed by $x$ in the tree learned on $\train$.
Intuitively, if we compute the set of all traces $\tlearner(\train,x)$ for each $x \in \train$,
we get the full tree, the one that we would have traditionally learned for $\train$.

The learner $\tlearner$ repeats two core operations: 
\rone selecting a predicate $\varphi$
with which to split the dataset (using $\bestsplit$)
and 
\rtwo removing elements of the training set
based on whether they satisfy the predicate $\varphi$ (depending on $x$,
using $\filter$).\footnote{%
Note that \rtwo is what distinguishes our trace-based learning
from conventional learning of a full tree.
\rone selects predicates that would comprise the tree, while
\rtwo directs us to recurse \emph{only} along the path that the specific $\point$ would take,
as opposed to recursing down both
(and without affecting how predicates in the tree are selected).
}
The number of times the loop is repeated ($d$) is the maximum depth of the trace that is constructed.
Throughout, we assume a fixed  set of classes $\classes = \{1,\ldots,k\}$.

The mutable state of $\tlearner$ is the triple  $(\train, \varphi, \seq)$:
\begin{itemize}[topsep=5pt, partopsep=0pt, leftmargin=*]
    \item $\train$ is the training set, which will keep getting refined (by dropping elements) as the trace is constructed.
    \item $\varphi$ is the most recent predicate along the trace, which is initially undefined (denoted by $\nullpred$).
    \item $\seq$ is the sequence of predicates along the trace, which is initially empty.
\end{itemize}

\begin{figure}
\begin{flushleft}
\textbf{Input: } training set $\train$ and input $\point \in \tdom$
~\\
\textbf{Initialize: } $\varphi \gets \nullpred$, 
$\seq \gets \text{ empty trace}$
~\\[1mm]
repeat $d$ times
~\\
\hspace{1em}$\ifs \ \impurity(\train) = 0 \ \thens\ \mathsf{return}$\\
\hspace{1em}$\varphi \gets \bestsplit(\train)$\\
\hspace{1em}$\ifs \ \varphi = \nullpred \ \thens \ \mathsf{return}$\\
\hspace{1em}$\train \gets \filter(\train,\varphi,\point)$\\
\hspace{1em}$\ifs \  \point \models \varphi \  \thens \ \seq \gets \seq\varphi  \ \elses \  \seq \gets \seq\neg\varphi$
~\\[1mm]
\hspace{0em}\textbf{Output:} $\argmax_{i \in [1,k]} p_i, \text{ where } \summary(\train)=\langle p_1,\ldots,p_k\rangle
$
\end{flushleft}
\caption{Trace-based decision-tree learner $\tlearner$}\label{fig:alg}
\end{figure}

\paragraph{Predicate Selection}
We assume that $\tlearner$ is equipped with a finite set of predicates $\Phi$
with which it can construct a decision-tree classifier;
 each predicate in $\Phi$ is a Boolean function in $\tdom \rightarrow \mathbb{B}$.

$\bestsplit(\train)$ computes a predicate $\phibest \in \Phi$ that splits the current dataset $\train$---usually minimizing a notion of
entropy.
Ideally, the learning algorithm would consider every possible sequence of predicates to partition a dataset
in order to arrive at an optimal classifier.
For efficiency, a decision-tree-learning  algorithms does this greedily: it selects the best predicate it can find for a single split and moves on to the next split.
To perform this greedy choice, it measures how diverse the two datasets resulting from the split are.
We formalize this below:

We use $\restr{\train}{\varphi}$ to denote the subset of $\train$ that satisfies $\varphi$,
i.e., $$\restr{\train}{\varphi} = \{(x,y) \in \train \mid x \models \varphi\}$$
%
Let $\Phi'$ be the set of all predicates
that do not trivially split the dataset:
$\Phi' = \{\varphi \in \Phi \mid \restr{\train}{\varphi} \neq \emptyset \wedge \restr{\train}{\varphi} \neq\train \}$.
Finally, $\bestsplit(\train)$ is defined as follows:
\[
\bestsplit(\train) = 
    \argmin_{\varphi\in\Phi'}\ \score(\train,\varphi) \]
where $\score(\train,\varphi) = |\restr{\train}{\varphi}| \cdot \impurity(\restr{\train}{\varphi}) + |\restr{\train}{\neg\varphi}| \cdot \impurity(\restr{\train}{\neg\varphi})$.
Informally, $\bestsplit(\train)$ is the predicate 
that splits $\train$ into two sets with the lowest 
entropy, as defined by the function $\impurity$ shown in Figure~\ref{fig:aux}.
Formally, $\impurity$ computes \emph{Gini impurity}, which is used, for instance,
in the CART algorithm~\cite{breiman2017classification}.
Note that if  $\Phi' = \emptyset$,
we assume $\bestsplit(\train)$ is undefined (returns $\nullpred$).
Further, if multiple predicates are possible values of $\bestsplit(\train)$,
we assume one is returned nondeterministically.
Later, in Section~\ref{sec:absdomain}, our abstract interpretation of $\tlearner$ will actually  capture all possible predicates in the case of a tie.

\begin{example}
\label{ex:entropy-cprob}
    Recall our example from Section~\ref{sec:overview}
    and Figure~\ref{fig:example}.
    For readability, we use $\train$ instead of $T_{bw}$ for the name of the dataset.
    Let us compute $\score(T,\varphi)$, where $\varphi$ is $x \leq 10$.    
    We have $|\restr{\train}{\varphi}| = 9$ and
    $|\restr{\train}{\neg\varphi}| = 4$.
    For the classification probabilities, defined by $\summary$ (Figure~\ref{fig:aux}), we have
    $\summary(\restr{\train}{\varphi}) = \langle 7/9, 2/9 \rangle$ and 
    $\summary(\restr{\train}{\neg\varphi}) = \langle 0, 1 \rangle$
    assuming the first element represents white classification;
    e.g., in $\restr{\train}{\varphi}$, there's a $7/9$ chance of being classified as white.
    For $\impurity$, we have
    $\impurity(\restr{\train}{\varphi}) \approx 0.35$ and
    $\impurity(\restr{\train}{\neg\varphi}) = 0$.
    Since $\restr{\train}{\varphi}$ is solely
    composed of black points, its Gini impurity is 0.

    The score of $x \leq 10$ is therefore ${\sim}3.1$.
    For the predicate $x \leq 11$, we get the higher (worse)
    score of ${\sim}3.2$,
    as it generates a more diverse split.

\end{example}

\begin{figure}
\centering
\begin{align*}
    \impurity(\train) &= \sum_{i=1}^k p_i(1-p_i), \ \text{ where } \summary(\train) = \langle p_1,\ldots,p_k \rangle
\\
    \summary(\train) &= \left \langle\frac{|\{(\point,y) \in \train \mid y = i \}|} { |\train| } \right \rangle_{i \in [1,k]}
\end{align*}
\caption{Auxiliary operator definitions. $\impurity$ is Gini impurity;
$\summary$ returns a vector of classification probabilities, one element for each class $i \in [1,k]$.}
\label{fig:aux}
\end{figure}

\paragraph{Filtering the Dataset}
The operator $\filter$ removes elements of $\train$ that 
evaluate differently than $\point$ on $\varphi$.
Formally, 
\[
\filter(\train,\varphi,\point) =\begin{cases}
    \restr{\train}{\varphi} & \text{if } x \models \varphi\\
    \restr{\train}{\neg\varphi} & \text{otherwise}
    \end{cases}
\]

\paragraph{Learner Result}
When $\tlearner$ terminates in a state \linebreak $(\train_r,\varphi_r,\seq_r)$,
we can read the classification of $x$ as the class $i$ with the highest 
number of training elements in $\train_r$.

Using $\summary$, in Figure~\ref{fig:aux},
we compute the probability of each class $i$
for a training set $\train$ as a vector of probabilities.
Finally, $\tlearner$ returns the class with the highest probability:
\[\argmax_{i \in [1,k]} p_i\hspace{1cm} \text{ where } \summary(\train_r)=\langle p_1,\ldots,p_k\rangle
\]
As before, in case of a tie in probabilities, we assume a nondeterministic choice.

\begin{example}
\label{ex:training-example}
    Following the computation from Ex.~\ref{ex:entropy-cprob}, \linebreak
    $\tlearner(T,18)$ terminates in  state 
    $(\restr{\train}{x > 10}, \ {x \leq 10}, \ [x > 10])$.
    Point 18 is associated with the trace $[{x >  10}]$ 
    and is classified as black because     
    $\summary(\restr{\train}{x > 10}) = \langle 0, 1 \rangle$. 
\end{example}


%% file: absdomain.tex
\section{Abstractions of Poisoned Semantics} \label{sec:absdomain}
In this section, we begin by defining a data-poisoning model in which an attacker contributes a number of malicious training items.
Then, we demonstrate how to apply the trace-based learner $\tlearner$
to \emph{abstract sets of training sets}, allowing us to
efficiently prove poisoning-robustness.

\subsection{The $n$-Poisoning Model}\label{sec:npoison-model}
For our purposes, we will consider a poisoning model
where the attacker has contributed up to $n$ elements of the training set---%
we call it $n$-\emph{poisoning}.
%
%
Formally, given a training set $\train$ and a natural number
$n \leq |\train|$, 
we define the following perturbed set:
\[
\Drop_n(\train) = \{\train' \subseteq \train \ : \ |\train \setminus \train'| \leq n\}
\]
In other words,
$\Drop_n(\train)$ captures every training set
the attacker could have possibly started from to arrive at $\train$.

This definition of dataset poisoning matches many 
settings studied in the literature~\cite{chen2017targeted,steinhardt2017certified,xiao2015feature}.
The idea is that an attacker has contributed a number of malicious
data points into the training set to influence the resulting classifier.
For example, \citet{chen2017targeted} consider poisoning a facial recognition model to enable bypassing authentication,
and \citet{xiao2015feature} consider poisoning a malware detector to allow the attacker to install malware.

We do not know which $n$ points in $\train$
are the malicious ones, or if there are malicious points at all. Thus, the set $\Drop_n(\train)$
captures every possible subset of $\train$ where we have removed up to $n$ (potentially malicious) elements.
Our goal is to prove that our classification is robust to up to $n$
possible poisoned points added by the attacker.
So if we try every possible dataset in $\Drop_n(\train)$ and they all result in the same classification on $\point$, then $x$ is robust regardless of the attacker's potential contribution.

Observe that $|\Drop_n(\train)| = \sum_{i=1}^n \binom{|\train|}{i}$.
So even for relatively small datasets and number $n$,
the set of possibilities is massive, e.g., 
for MNIST-1-7 dataset (\S\ref{sec:evaluation}),
for $n=50$, we have about $10^{141}$ possible training sets in $\Drop_n(\train)$.

\subsection{Abstract Domains for Verifying $n$-Poisoning}
\label{sec:abstract-domains-npoisoning}
Our goal is to efficiently evaluate $\tlearner$
on an input $\point$ for all possible training datasets in $\Drop_n(\train)$. If all of them yield the same classification $\class$,
then we know that $\point$ is a robust input.
Our insight is that we can abstractly interpret $\tlearner$
on a symbolic set of training sets without having to fully 
expand it into all of its possible concrete instantiations.
This allows us to train on an enormous number of datasets,
which would be impossible via enumeration. 

Recall that the state of $\tlearner$ is $(\train, \varphi, \seq)$;
for our purposes, we do not have to consider the sequence of predicates $\seq$, as we are only interested in the final classification, which is a function of $\train$.
In this section, we present the \emph{abstract domains} for each component of the learner's state.


\paragraph{Abstract Training Sets}
Abstracting training sets is the main novelty of our technique.
We use the  abstract element $\drop{\train'}{n'}$ to denote a set of training sets
and it captures the definition of 
$\Drop_{n'}(\train')$:
For every training set $T'$ and number $n'$, the concretization
function is $\gamma\left(\drop{\train'}{n'}\right)=\Drop_{n'}(\train')$.
Therefore, we have that initially the \emph{abstraction} function 
 $\alpha(\Drop_n(\train))=\drop{\train}{n}$ is precise.
 Note that an abstract element $\drop{\train'}{n'}$
 succinctly captures a large number of concrete sets, $\Drop_{n'}(\train')$.
 Further, all operations we perform on 
 $\drop{\train'}{n'}$ will only modify
 $\train'$ and $n'$, without resorting to concretization.

We can define an \emph{efficient} join operation on two elements in the abstract domain%
\footnote{%
Elements in the domain are ordered so that $\drop{\train_1}{n_1} \sqsubseteq \drop{\train_2}{n_2}$
if and only if $\train_1 \subseteq \train_2 \land n_1 \leq n_2 - |\train_2 \setminus \train_1|$.
In the text, we define the concretization function,
a special case of the abstraction function, and the join operation;
note that we do not require an explicit meet operation for the purposes of this paper---%
although one is well-defined:
\begin{align*}
\drop{\train_1}{n_1} \sqcap \drop{\train_2}{n_2} \coloneqq&~
\text{if } |\train_1 \setminus \train_2| > n_1 \lor |\train_2 \setminus \train_1| > n_2
\text{ then } \bot \\
&~\text{else } \drop{\train_1 \cap \train_2}{\min(n_1 - |\train_1 \setminus \train_2|, n_2 - |\train_2 \setminus \train_1|)}
\end{align*}
}
as follows:
\begin{definition}[Joins]
\label{def:join}
Given two training sets $\train_1, \train_2$ and $n_1,n_2 \in \mathbb{N}$,
$\drop{\train_1}{n_1} \sqcup \drop{\train_2}{n_2} \coloneqq
\drop{\train'}{n'}$
where $\train' = \train_1 \cup \train_2$
and $n' = \max(|\train_1 \setminus \train_2| + n_2, |\train_2 \setminus \train_1| + n_1)$.
\end{definition}
Notice that the join of two sets is an overapproximation of the union of the two sets.
The following proposition formalizes the soundness of this operation:
\begin{proposition}\label{prop:join}
For any $\train_1$, $\train_2$, $n_1$, $n_2$, the following holds:
\[\gamma\left(\drop{\train_1}{n_1}\right) \cup \gamma\left(\drop{\train_2}{n_2}\right)
\subseteq \gamma(\drop{\train_1}{n_1} \sqcup \drop{\train_2}{n_2}).
\]
\end{proposition}

\begin{example}
For any training set $\train_1$, if we consider the abstract sets
$\drop{\train_1}{2}$ and $\drop{\train_1}{3}$, because
the second set represents strictly more
concrete training sets, we have
$$\drop{\train_1}{2} \sqcup \drop{\train_1}{3}
=
\drop{\train_1}{3}$$

Now consider the training set $\train_2 = \{x_1,x_2\}$.
We have 
\[
    \drop{\train_2}{2} \sqcup \drop{\train_2\cup\{x_3\}}{2}
= \drop{\train_2\cup\{x_3\}}{3}
    \]
Notice how the join increased the poisoned elements from 2 to 3 to accommodate for the additional element $x_3$.

\end{example}

\paragraph{Abstract Predicates and Numeric Values}
When abstractly interpreting what predicates the learner might choose for different training
sets, we will need to abstractly represent sets of possible predicates.
Simply, a set of predicates is abstracted \emph{precisely} as the corresponding set of predicates $\Psi$---i.e., 
for every set $\Psi$, we have $\alpha(\Psi)=\Psi$
and $\gamma(\Psi)=\Psi$. Moreover, $\Psi_1\sqcup \Psi_2=\Psi_1\cup \Psi_2$.
For certain operations, it will be handy for $\Psi$ to contain a special null predicate $\nullpred$.

When abstractly interpreting numerical operations, like $\summary$ and $\impurity$,
we will need to abstract sets of numerical values.
We do so using the standard
\emph{intervals} abstract domain (denoted $[l,u]$).
For instance,
$\alpha(\{0.2,0.4,0.6\})=[0.2,0.6]$
and 
$\gamma([0.2,0.6])=\{x\mid 0.2 \leq x \leq 0.6\}$.
The join of two intervals is defined as $[l_1,u_1]\sqcup [l_2,u_2]=[\min(l_1,l_2),\max(r_1,r_2)]$.
Interval arithmetic follows the standard definitions and we thus elide it here.%
\footnote{While we choose intervals as our numerical abstract
domain in this paper, any numerical abstract domain could be used.}

\subsection{Abstract Learner $\atlearner$}
We are now ready to define an abstract interpretation of the semantics of our decision-tree learner, denoted $\atlearner$.

\paragraph{Abstract Domain}
Recall that the state of $\tlearner$ is $(\train, \varphi, \seq)$;
for our purposes, we do not have to consider the sequence of predicates $\seq$, as we are only interested in the final classification, which is a function of $\train$.
Using the domains described in Section~\ref{sec:abstract-domains-npoisoning},
at each point in the learner, our abstract state is 
a pair
$(\drop{\train'}{n'}, \Psi')$ (i.e., in the product abstract domain) that tracks the current set of training sets and the current set of possible most recent predicates
the algorithm has split on (for all considered training sets).

When verifying \textit{n}-poisoning for a training set $\train$, the initial abstract state of the learner will be the pair
$(\drop{\train}{n}, \{\nullpred\})$.
In the rest of the section, we define the abstract semantics (i.e., our abstract transformers) for all 
the operations performed by $\atlearner$.
For operations that only affect one element of the state, we assume that the other component is left unchanged.

\subsection{Abstract Semantics of Auxiliary Operators}\label{sec:auxops}
We will begin by defining the abstract semantics
of the auxiliary operations in the algorithm before proceeding
to the core operations, $\filter$ and $\bestsplit$.
This is because the auxiliary operators are simpler and highlight
the nuances of our abstraction.

Let us begin by considering $\arestr{\drop{\train}{n}}{\varphi}$,
which is the abstract analog of $\restr{\train}{\varphi}$.
\begin{equation}
\label{eq:abstract-semantics-dropn}
\arestr{\drop{\train}{n}}{\varphi} \coloneqq \drop{\restr{\train}{\varphi}}{\min(n,|\restr{\train}{\varphi}|)}
\end{equation}
Simply, it removes elements not satisfying $\varphi$ from $\train$;
since the size of $\restr{\train}{\varphi}$ can go below $n$, we take the minimum of the two.

\begin{proposition}\label{prop:arestr}
    Let $\train' \in \gamma(\drop{\train}{n})$.
    For any predicate $\varphi$, we have $\restr{\train'}{\varphi} \in \gamma(\arestr{\drop{\train}{n}}{\varphi})$.
\end{proposition}
Now consider $\summary(\train)$, which returns a vector of probabilities for different classes.
Its abstract version returns an interval for each probability,
denoting the lower and upper bounds based on the training sets in the abstract set:%
\footnote{%
Note that this transformer can be more precise:
for example, the interval division as written is not guaranteed to be a subset of $[0,1]$,
despite the fact that all concrete values would be.
Throughout this section, many of the transformers are simply the ``natural'' lifting
of numerical arithmetic to interval arithmetic;
while this may not be optimal, we do so to make it easier
to see the correctness of the approach (and to make proofs and implementation straightforward).

In the case of $\asummary$, we can compute the optimal transformer inexpensively:
it is equivalent to write that $\summary(\train)$ computes, for each class $i \in [1,k]$,
the \emph{average} of the multiset
${S_i = [\text{if } \class=i \text{ then } 1 \text{ else } 0 \mid (\point,\class) \in \train]}$.
We can then have $\asummary(\drop{\train}{n})$ perform a similar computation for each component:
let $L_i$ denote the $m$-many least elements of $S_i$,
and let $U_i$ denote the $m$-many greatest elements of $S_i$,
where $m = |\train| - n$.
These $L_i$ and $U_i$ exhibit extremal behavior of averaging,
so we can directly compute the endpoints of the interval assigned to each class
as $[\frac{1}{m} \sum_{b\in L_i} b, \frac{1}{m} \sum_{b\in U_i} b]$.

Note that our implementation used for the evaluation (Section~\ref{sec:evaluation})
\emph{does} employ this optimal transformer for $\asummary$,
while the other transformers match what is presented.
}%
\[
\asummary(\drop{\train}{n}) \coloneqq
\left\langle \frac{[\max(0,c_i-n), c_i]}{[|T| - n, |T|]} \right\rangle_{i \in [1,k]}
\]
where $c_i = |\{(x,i) \in \train\}|$.
In other words, for each class $i$, we need to consider
the best- and worst-case probability based on removing $n$ elements from the training set, as denoted by the denominator and the numerator.
Note that in the corner case where $n=|\train|$,
we set $\asummary(\drop{\train}{n}) = \langle[0,1]\rangle_{i\in[1,k]}$.

\begin{proposition}\label{prop:asummary}
    Let $\train' \in \gamma(\drop{\train}{n})$.
    Then, $$\summary(\train') \in \gamma\left(\asummary(\drop{\train}{n})\right)$$
    where $\gamma\left(\asummary(\drop{\train}{n})\right)$ is the set of all possible probability vectors in the vector of intervals.
\end{proposition}

\begin{example}
Consider the training set on the left side of the tree in Figure~\ref{fig:example}; call it $\train_\ell$.
It has 7 white elements and 2 black elements.
$\summary(\train_\ell)=\left\langle 7/9, 2/9\right\rangle$,
where the first element is the white probability.
$\asummary(\drop{\train_\ell}{2})$
produces the vector
$\left\langle \left[5/9, 1\right], \left[0, 2/7\right]\right\rangle$.
Notice the loss of precision in the lower bound
of the first element.
If we remove two white elements, we should get a probability of $5/7$, but the interval domain cannot capture the relation between the numerator and denominator in the definition of  $\asummary$.
\end{example}
The abstract version of the  Gini impurity is identical to the concrete one,
except that  it performs interval arithmetic:
\[
    \aimpurity(\train) = \sum_{i=1}^k \iota_i([1,1]-\iota_i), \ \text{ where } \asummary(\train) = \langle\iota_1,\ldots,\iota_k\rangle    
\]
Each term $\iota_i$ denotes an interval.

\subsection{Abstract Semantics of $\filter$}
We are now ready to define the abstract version of $\filter$.
Since we are dealing with abstract training sets, as well as a set of predicates 
$\Psi$, we need to consider for each $\varphi\in\Psi$
all cases where $x \models \varphi$ or $x \models \neg \varphi$, and take the join of all the resulting training sets (Definition~\ref{def:join}).
Let 
$$\Psi_x = \{\varphi \in \Psi \mid x \models \varphi\}
~\text{ and  }~ \Psi_{\neg x} = \{\varphi \in \Psi \mid x \models \neg\varphi\}$$
Then,
\[
    \afilter(\drop{\train}{n}, \Psi, \point) \coloneqq
    \left(\bigsqcup_{\varphi \in \Psi_x} \arestr{\drop{\train}{n}}{\varphi}\right)
    \sqcup \left(\bigsqcup_{\varphi \in \Psi_{\neg x}} \arestr{\drop{\train}{n}}{\neg\varphi}\right)
\]

\begin{proposition}\label{prop:afilter}
    Let $\train' \in \gamma(\drop{\train}{n})$ and $\varphi' \in \Psi$.
    Then,
    $$
    \filter(\train',\varphi',\point) \in
    \gamma\left(\afilter(\drop{\train}{n},\Psi,\point)\right)
    $$
\end{proposition}

\begin{example}
Consider the full dataset $T_{\emph{bw}}$ from Figure~\ref{fig:example}.
For readability, we write $T$ instead of $T_{\emph{bw}}$ in the example.
Let $x$ denote the input with numerical feature 4,
and let $\Psi = \{x \leq 10\}$.
First, note that  because
$\Psi_{\neg x}$ is the empty set, the right-hand side of the result of applying the $\afilter$
operator will be the bottom element $\drop{\emptyset}{0}$ (i.e., the identity element for $\sqcup$).
Then, 
\begin{align*}\afilter(\drop{\train}{2},\Psi,x) &= 
\arestr{\drop{\train}{2}}{x\leq 10} \sqcup \drop{\emptyset}{0} & (\text{def. of } \afilter)\\
&= \drop{\restr{\train}{x \leq 10}}{2} \sqcup \drop{\emptyset}{0}& (\text{def. of } \arestr{\drop{\train}{n}}{\varphi})\\
&= \drop{\restr{\train}{x \leq 10}}{2} & (\text{def. of } \sqcup).
\end{align*}
%
\end{example}

\subsection{Abstract Semantics of $\bestsplit$}
We are now ready to define the abstract version of $\bestsplit$.
We begin by defining $\abestsplit$ without handling trivial predicates,
then we refine our definition.

\paragraph{Minimal Intervals}
Recall that in the concrete case, $\bestsplit$
 returns a predicate 
that minimizes the function $\score(\train, \varphi)$.
To lift $\bestsplit$ to the abstract semantics,
we define $\ascore$, which returns an interval, and what it means to be a \emph{minimal} interval---i.e., the interval
corresponding to the abstract minimal value of the objective function
$\ascore(\train, \varphi)$.

Lifting $\score(\train, \varphi)$ to  $\ascore(\drop{\train}{n}, \varphi)$ can
be done using the
sound transformers for the intermediary computations:
\begin{align*}
\ascore(\drop{\train}{n}, \varphi) \coloneqq&~
|\arestr{\drop{\train}{n}}{\varphi}| \cdot \aimpurity(\arestr{\drop{\train}{n}}{\varphi}) \\
&~+ |\arestr{\drop{\train}{n}}{\neg\varphi}| \cdot \aimpurity(\arestr{\drop{\train}{n}}{\neg\varphi})
\end{align*}
where $\lvert\drop{\train}{n}\rvert \coloneqq [|\train| - n, |\train|]$.

However, given a set of predicates $\Phi$,
$\abestsplit$ must return the ones with the minimal scores.
Before providing the formal definition, we illustrate the idea with an example.
\begin{example}
Imagine a set of predicates $\Phi=\{\varphi_1,\varphi_2,\varphi_3,\varphi_4\}$
with the following intervals for $\ascore(\drop{\train}{n},\varphi_i)$.
\begin{figure}[h!]
  \includegraphics[scale=1.5]{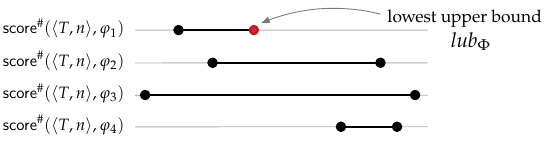}
\end{figure}

\noindent
Notice that $\varphi_1$ has the lowest upper bound for score (denoted in red and named $\lub_{\Phi}$).
Therefore, we call $\ascore(\drop{\train}{n},\varphi_1)$ the \emph{minimal interval} with respect to $\Phi$.
$\abestsplit$ returns all the predicates whose scores \emph{overlap} with the minimal interval $\ascore(\drop{\train}{n},\varphi_1)$, 
which in this case are $\varphi_1$, $\varphi_2$, and $\varphi_3$. 
This is because  there is a chance that $\varphi_1$, $\varphi_2$ and $\varphi_3$ are indeed
the predicates with the best score, but our abstraction has lost too much
precision for us to tell conclusively. 
\end{example}

Let $\lb$/$\ub$ be functions that return the lower/upper bound of an interval.
First, we define the lowest upper bound among the abstract scores of the predicates in $\Phi$ as 
\[
\lub_\Phi = \min_{\varphi\in\Phi} \ub(\ascore(\drop{\train}{n},\varphi))
\]
We can now define the set of predicates whose score overlaps with the minimal interval as:
\[
\{ \varphi \in \Phi \mid \lb(\ascore(\drop{\train}{n},\varphi)) \leq \lub_\Phi \}
\]

\paragraph{Dealing with Trivial Predicates}
Our formulation above considers the full set of predicates, $\Phi$.
To be more faithful to the concrete semantics,
$\abestsplit$ needs to eliminate trivial predicates from this set.
In the concrete case, we only considered $\varphi$
as a possible best split if $\varphi$ performed a non-trivial split on $\train$,
which we denoted $\varphi \in \Phi'$.
(Recall that a trivial split of $\train$ is one that returns $\emptyset$ or $\train$.)

This is a little tricky to lift to our abstract case,
since  a predicate $\varphi$ could non-trivially split
\emph{some} of the concrete datasets but not others.
We lift the set $\Phi'$ in two ways:
\begin{itemize}[topsep=5pt, partopsep=0pt, leftmargin=*]
\item {\it Universal predicates}:  the predicates that are non-trivial splits
\emph{for all} concrete training sets in $\gamma(\drop{\train}{n})$\footnote{Note that checking  $\emptyset \not\in \gamma(\drop{\train}{n})$
is equivalent to checking $n \neq |\train|$.}
\[    \Phi_\forall = \{\varphi \in \Phi \mid
        \emptyset \not\in \gamma(\arestr{\drop{\train}{n}}{\varphi}) \land
        \emptyset \not\in \gamma(\arestr{\drop{\train}{n}}{\neg\varphi})\}
\]

\item {\it Existential predicates}:  the predicates that are non-trivial splits
\emph{for at least one} concrete training set in $\gamma(\drop{\train}{n})$
\[
\Phi_\exists = \{\varphi \in \Phi \mid
    \drop{\emptyset}{\cdot} \neq \arestr{\drop{\train}{n}}{\varphi} \land
    \drop{\emptyset}{\cdot} \neq \arestr{\drop{\train}{n}}{\neg\varphi}\}
\]
\end{itemize}
Finally, the definition of $\abestsplit$ considers two cases:
\begin{align*}
    \abestsplit(\drop{\train}{n}) \coloneqq&~
    \text{if } \Phi_\forall = \emptyset \text{ then } \Phi_\exists \cup \{\nullpred\} \text{ else } \\
    &~ \{\varphi \in \Phi_\exists : \lb(\ascore(\drop{\train}{n}, \varphi)) \leq \lub_{\Phi_\forall}\}
\end{align*}

 

The first case captures when no single predicate is non-trivial for all sets:
we then return all predicates that succeed on at least one training set in $\drop{\train}{n}$,
since we cannot be sure one is strictly better than another.
To be sound, we also assume the cause of $\Phi_\forall$ being empty is 
a particular concrete training set for which every predicate forms a trivial split,
hence we include $\nullpred$ as a possibility.
The second case corresponds to returning the predicates with minimal scores.

%

\begin{lemma}\label{lem:abestsplit}
Let $\train' \in \gamma(\drop{\train}{n})$.
Then, 
\[
    \bestsplit(\train') \in
 \gamma(\abestsplit(\drop{\train}{n}))
 \]
\end{lemma}

\subsection{Abstracting Conditionals}
We abstractly interpret conditionals in $\tlearner$,
as is standard, by taking the join of all abstract states
from the feasible \emph{then} and \emph{else} paths.
In $\tlearner$, there are two branching statements of interest for our purposes,
one with the condition $\impurity(\train) = 0$ and one with $\varphi = \nullpred$. 

Let us consider the condition $\varphi = \nullpred$.
Given an abstract state $(\drop{\train}{n},\Psi)$,
we simply set $\Psi = \{\nullpred\}$ and propagate the state to the \emph{then} branch
(unless, of course, $\nullpred \not\in \Psi$, in which case we omit this branch).
For $\varphi \neq \nullpred$, we remove $\nullpred$ from $\Psi$ and propagate the resulting state through the \emph{else} branch.

Next, consider the conditional $\impurity(\train) = 0$.
For the \emph{then} branch,
we need to \emph{restrict}
an abstract state $(\drop{\train}{n},\Psi)$
to training sets with 0 entropy:
intuitively, this occurs when all elements have the same classification.
We ask: \emph{are there any concretizations composed of elements of the same class?},
and we proceed through the \emph{then} branch with the following training set abstraction:
\[
\bigsqcup_{i \in [1,k]} \puresets(\drop{\train}{n},i)
\]
where
\begin{align*}
\puresets(\drop{\train}{n},i) \coloneqq&~
\text{Let } T' = \{(x,y) \in \train \mid y=i\} \text{ in } \\
&~
\text{if } |\train \setminus \train'| \leq n
\text{ then } \drop{\train'}{n- |\train\setminus \train'| } \\
&~
\text{else } \bot
\end{align*}
The idea is as follows:
the set $\train'$ defines a subset of $\train$ containing
only elements of class $i$.
But if we have to remove more than $n$ elements from $\train$
to arrive at $\train'$, then the conditional is not realizable
by a concrete training set of class $i$,
and so we return the empty abstract state.

In the case of $\impurity(\train) \neq 0$ (the \emph{else} branch),
we soundly (imprecisely) propagate the original state without restriction.

\subsection{Soundness of Abstract Learner}
Finally, $\atlearner$
soundly overapproximates the results of  $\tlearner$ and can therefore be used to
prove robustness to $n$-poisoning.

\begin{theorem}\label{thm:main}
Let $\train' \in \gamma(\drop{\train}{n})$,
let $(\train_f', \cdot,\cdot)$ be the final state of $\tlearner(\train',\point)$,
and let $(\drop{\train_f''}{n_f},\cdot)$ be the final abstract state of $\atlearner(\drop{\train}{n},\point)$.
Then $\train_f' \in \gamma(\drop{\train_f''}{n_f})$.
\end{theorem}

It follows from the soundness of $\atlearner$ 
that we can use it to prove $n$-poisoning robustness.
Let $I=([l_1,u_2], \ldots, [l_k,u_k])$ be a set of intervals.
We say that interval $[l_i, u_i]$ dominates $I$ 
if and only if $l_i > u_j$ for every $j\neq i$.

\begin{corollary}
Let $\drop{\train'}{n'}$ be the final abstract state of 
$\atlearner(\drop{\train}{n},\point)$.
If
$I=\asummary(\drop{\train'}{n'}))$ and there exists 
an interval in $I$ that dominates $I$ (i.e.,  same class is selected for every $\train\in\gamma\drop{\train}{n}$),
then 
$\point$ is robust
to $n$-poisoning of $\train$.
\end{corollary}


%% file: extensions.tex
\section{Extensions}
\label{sec:extensions}

In this section, we present two extensions that make our abstract interpretation 
framework more practical.
First, we show how our abstract domain can be modified  to accommodate
real-valued features (\S~\ref{sec:real-valued-features}).
Second, we present a disjunctive abstract domain that is more precise than the one
we discussed, but more computationally inefficient (\S~\ref{sec:disjunctive-abstraction}).

\subsection{Real-Valued Features}
\label{sec:real-valued-features}

Thus far, we have assumed that $\tlearner$ and $\atlearner$ operate on
a finite set of predicates $\Phi$.
In real-world decision-tree implementations, this is not quite accurate:
for real-valued features, there are infinitely many possible predicates of the form
$\lambda \point_i \ldotp \point_i \leq \tau$ (where $\tau \in \mathds{R}$),
and the learner chooses a finite set of possible $\tau$ values dynamically,
based on the training set $\train$.
We will use the subscript $\reals$ to denote the real-valued
versions of existing operations.

\paragraph{From $\tlearner$ to $\tlearnerr$}
The new learner $\tlearnerr$ is almost identical to $\tlearner$.
However, each invocation of $\bestsplitr$
first computes a finite set of predicates $\Phi_\reals$.
%
Consider all of the values appearing in $\train$ for the $i$th feature in $\tdom$,
sorted in ascending order.
For each pair of adjacent values $(a, b)$ (i.e., such that there exists no $c$ in $\train$
such that $a<c<b$),
we include in $\Phi_\reals$ the predicate $\varphi=\lambda \point_i \ldotp \point_i \leq \frac{a + b}{2}$.

\begin{example}\label{ex:disj}
In our running example from Figure~\ref{fig:example},
we have training set elements in $\train_\emph{bw}$
whose features take the numeric values $\{0,1,2,3,4,7,\ldots,14\}$.
$\bestsplitr(\train_\emph{bw})$ would pick a predicate from the set
$\Phi_\reals = \{\lambda \point \ldotp \point \leq \tau \mid
\tau \in \{\frac{1}{2},\frac{3}{2},\frac{5}{2},\frac{7}{2},\frac{11}{2},\frac{15}{2},\ldots,\frac{27}{2}\} \}$.
\end{example}

\paragraph{From $\atlearner$ to $\atlearnerr$}
To apply the abstract learner in the real-valued setting,
we can follow the idea above and construct a finite set $\Phi_\reals$.
Because our poisoning model assumes dropping up to $n$ elements of the training set,
this results in roughly $(n+1)\cdot|T|$ predicates in the worst case---i.e.,
we need to account for every pair $(a,b)$ of adjacent feature values 
or that are adjacent after removing up to $n$ elements between them.

\begin{example}
    Continuing Example~\ref{ex:disj}.
    Say we want to compute $\Phi_\reals$ for $\drop{\train_\emph{bw}}{1}$.
    Then, for every pair of values that are $1$ apart
    we will need to add a predicate to accommodate the possibility that we
    drop the value between them.
    E.g., in $\train_\emph{bw} = \{\ldots,3,4,7,\ldots\}$,
     we will additionally need the predicate $\lambda x \ldotp x \leq (3+7)/2$,
    for the case where we drop the element with value $4$  from the dataset.
\end{example}

To avoid a potential explosion in the size of the predicate set 
and maintain efficiency, we compactly represent
sets of similar predicates symbolically.
We describe this detail \iffull in Appendix~\ref{app:real-valued-features}\else ~in the
full version of our paper~\cite{antidotearxiv}\fi.

\subsection{Disjunctive Abstraction}
\label{sec:disjunctive-abstraction}



The state  abstraction used by $\atlearner$
can be \emph{imprecise},
mainly due to the join operations that take place, e.g.,
during $\afilter$.
The primary concern is that we are forced to perform a very imprecise join
between possibly quite dissimilar training set fragments.
Consider the following example:

\begin{example}
Let us return to $\train_\emph{bw}$ from Figure~\ref{fig:example},
but imagine we have continued the computation after filtering using $\point \leq 10$
and have selected some best predicates.
Specifically, consider a case in which we have
$\point=4$ and 
\begin{itemize}[topsep=5pt, partopsep=0pt, leftmargin=*]
\item $\drop{\train}{1}$, where $\train=\{0,1,2,3,4,7,8,9,10\}$
\item $\Psi = \{\point \leq 3, \point \leq 4\}$ (ignoring whether this is correct)
\end{itemize}
Let us evaluate $\afilter(\drop{\train}{1}, \Psi, \point)$.
Following the definition of $\afilter$, we will compute
$$\drop{\train'}{n'} =  \drop{\train_{\leq 4}}{1} \sqcup \drop{\train_{> 3}}{1}$$
where 
$\train_{\leq 4} = \{(4,b), (3,w), (2,w), (1,w), (0,b)\}$
and
$\train_{> 3} = \{(4,b), (7,w), (8,w), (9,w), (10,w)\}$,
thus giving us $\train' = \train$ (the set we began with)
and $n' = 5$ (much larger than what we began with).
This is a large loss in precision.
\end{example}

To address this imprecision, we will consider a 
\emph{disjunctive} version of our abstract domain, consisting
of unboundedly many disjuncts of this previous domain,
which we represent as a set $\{(\drop{\train}{n}_i, \Psi_i)\}_i$.
Our join operation becomes very simple: it is the union of the two sets of disjuncts.
\begin{definition}[Joins]
Given two disjunctive abstractions $D_I = \{(\drop{\train}{n}_i, \Psi_i)\}_{i\in I}$
and $D_J = \{(\drop{\train}{n}_j, \Psi_j)\}_{j\in J}$,
we define
\[
D_I \sqcup D_J \coloneqq D_I \cup D_J
\]
\end{definition}
Adapting $\atlearner$ to operate on this domain is immediate:
each of the transformers described in the previous section
is applied to each disjunct.

Because our disjunctive domain eschews memory- and time-efficiency for precision,
we are able to prove more things, but at a cost
(we explore this in our evaluation, \S~\ref{sec:evaluation}).
Note that, by construction, the disjunctive abstract domain is at least as precise
as our standard abstract domain.


%% file: evaluation.tex
\section{Implementation and Evaluation}\label{sec:evaluation}

We implemented our algorithms $\tlearner$ and $\atlearner$ in C++
in a (single-threaded) prototype we call \name.
Our evaluation%
\footnote{We use a machine with a 2.3GHz processor and 160GB of RAM throughout.}
aims to answer the following research questions:
\begin{description}
\item[RQ1] Can \name prove data-poisoning robustness for real-world datasets? (\S\ref{sec:effectiveness})
\item[RQ2] How does the performance of \name vary with respect to
    the scale of the problem and the choice of abstract domain?
    (\S\ref{sec:performance})
\end{description}

\subsection{Benchmarks and Experimental Setup}
We experiment on $\numBenchmarks$ datasets (Table~\ref{ta:benchmarks}).
We obtained the first three datasets
from the UCI Machine Learning Repository~\cite{UCI}.
\iris is a small dataset that categorizes three related flower species;
\mammography and \wdbc are two datasets of differing complexities
related to classifying whether tumors are cancerous.
We also evaluate on the widely-studied MNIST dataset of handwritten digits~\cite{mnist},
which consists of \numprint{70000} grayscale images (\numprint{60000} training, \numprint{10000} test)
of the digits zero through nine.
We consider a form of MNIST that has been used in the poisoning literature
and create another variant for evaluation:
\begin{itemize}[topsep=5pt, partopsep=0pt, leftmargin=*]
    \item We make the same simplification as in
        other work on data poisoning~\cite{Biggio12,steinhardt2017certified}
        and restrict ourselves to the classification of ones versus sevens
        (\numprint{13007} training instances and \numprint{2163} test instances),
        which we denote \mnistreal.
        \citet{steinhardt2017certified}, for example, 
        recently used this to study poisoning in support vector machines.
    \item Each \mnistreal image's pixels are 8-bit integers (which we treat as real-valued);
        to create a variant of the problem with reduced scale, we \textit{also} consider \mnistbin,
        a black-and-white version that uses each pixel's most significant bit (i.e.\ our predicates are Boolean).
\end{itemize}

\input{benchmarktable}

\input{data/firstfigure}

For each dataset, we consider a decision-tree learner
with a maximum tree depth (i.e.\ number of calls to $\bestsplit$)
ranging from 1 to 4.
Table~\ref{ta:benchmarks} shows that test set%
\footnote{The UCI datasets come as a single training set.
We selected a random 80\%-20\% split of the data,
saving the 20\% as the test set to use in our experiments.
The scale of the MNIST dataset is large;
for pragmatic reasons, we fix a random subset of 100 of the original \numprint{2163}
test set elements for robustness proving,
and we run our $\atlearner$ experiments only on this subset.}
accuracies of the decision trees learned by $\tlearner$ are reasonably high---%
affirmation that when we prove the robustness of its results,
we are proving something worthwhile.

\paragraph{Experimental Setup}
For each test element, we explore the amount of poisoning
(i.e.\ how large of a $n$ from our $\Drop_n$ model)
for which we can prove the robustness property as follows.
\begin{enumerate}[topsep=5pt, partopsep=0pt, leftmargin=*]
\item
For each combination of dataset $\train$ and tree depth $d$,
we begin with a poisoning amount $n=1$, i.e.\ a single element could be missing from the training set.
\item
For each test set element $x$, we attempt to prove that $x$ is robust to poisoning $\train$
using any set in $\Drop_n(T)$.
Let $S_n$ be the test subset for which we do prove robustness for poisoning amount $n$.
If $S_n$ is non-empty,
we double $n$ and again attempt to verify the property on elements in $S_n$.
\item If at a depth $n$ all instances fail, we binary search between $n$ and $n/2$ to find an $n/2<n'<n$ at which
some instances terminate. This approach allows us to better illustrate the experiment trends in our plots.
\end{enumerate}
Failure occurs due to any of three cases:
\rone the computed over-approximation does not conclusively prove robustness,
\rtwo the computation runs out of memory, or
\rthree the computation exceeds a one-hour timeout.
We run the entire procedure for the non-disjunctive and disjunctive abstract domains.

\subsection{Effectiveness of \name}
\label{sec:effectiveness}

We evaluate how effective \name is at proving data-poisoning robustness.
In this experiment, we consider a run of the algorithm on a single test element successful if
either the
non-disjunctive or disjunctive abstract domain  succeeds
at proving robustness
(mimicking a setting in which two instances of $\atlearner$, one for each abstract domain, are run in parallel)---%
we will contrast the results for the different domains in \S\ref{sec:performance}.
Figure~\ref{fig:mainresults} shows these results.

To exemplify the power of \name, draw your attention
to the depth-2 instance of $\atlearner$ invoked on \mnistreal.
For 38 of the 100 test instances, we are able to verify that
even if the training set had been poisoned by an attacker
who contributed up to 64 poisoned elements ($\approx \frac{1}{2}\%$),
the attacker would not have had any power to change the resulting classification.
Conventional machine learning wisdom says that,  in decision tree learning,
small changes to the training set can cause the model to behave quite differently.
Our results verify nuance---sometimes, there is some stability.\footnote{
The \iris dataset has an interesting quirk---%
we're unable to prove much at depth 1
because in the concrete case, one of the leaves
is a 50/50 split between two classes,
thus changing one element
could make the difference for any of the
test set instances taking that path.
At depth 2, a predicate is allowed to split that leaf further,
making decision-tree learning more stable.
}
These 38 verified instances average ${\sim}800s$ run time.
$\Drop_{64}(\train)$ consists of over $10^{174}$ concrete training sets;
This is staggeringly efficient compared to a na\"ive enumeration baseline,
which would be unable to verify robustness at this scale.


To answer \textbf{RQ1}, \textit{\name can verify  robustness for real-world datasets with extremely large
perturbed sets
and decision-tree learners with high accuracies.}

\subsection{Performance of \name}
\label{sec:performance}

We evaluate how the performance of \name is affected by the complexity of the problem,
e.g., the size of the training set and its number of features,
the number of poisoned elements, and the depth of the learned decision tree.
Due to the large number of parameters involved in our evaluation, this
section only provides a number of representative statistics.
In particular, although the reader can find plots describing all 
the metrics evaluated on each dataset \iffull in Appendix~\ref{app:bench}\else ~in the
full version of our paper~\cite{antidotearxiv}\fi,
most of our analysis will focus on \mnistbin~(see Figure~\ref{fig:mnistbin}), since
it exhibits the most illustrative behavior.

\paragraph{Box vs Disjuncts}
In this section we use Disjuncts to refer to the disjunctive abstract domain and Box to refer
to the non-disjunctive one.
Disjuncts is more precise than Box and, as expected, it can verify
more instances. However, Disjuncts is slower and more memory-intensive.
Consider the \mnistbin dataset~(see Figure~\ref{fig:mnistbin}).
For depth $3$ and $n=64$ (approximately $0.5\%$ of the dataset), 
Disjuncts can verify 52 instances while
Box can only verify 15. However,
Disjuncts takes on average 32s to terminate (0 timeouts) and uses 1,650MB of memory,
while
Box takes on average 0.7s to terminate (0 timeouts) and uses 150MB of memory.
It is worth noting that Box can verify certain instances that Disjunct cannot verify due to timeouts.
For example, at depth $4$ and $n=128$, Box is able to verify 1 problem instance,%
\footnote{The  14 other instances that succeeded at $n=64$
similarly terminated after 0.7s on average,
but their final state did not prove robustness.}
while Disjuncts always times out.
An interesting direction for future research would be to consider strategies
that capitalize on the precision of tracking many disjuncts
while incorporating the efficiency of allowing some to be joined.

\paragraph{Number of Poisoned Elements}
It is clear from the plots that the number of poisoned elements greatly
affects the performance and efficacy of \name. 
We do not focus on particular numbers, since the trends are clear from the plots \iffull (including the ones in Appendix~\ref{app:bench})\fi:
The memory consumption and running times of
Disjuncts grow exponentially with $n$, but are still practical and Disjuncts is effective up to high depths.
The memory consumption and running times of
Box grow more slowly:
95\% of all experiments we ran using Box finished within 20 seconds,
and none timed out (the longest took 232 seconds).%
\footnote{This data must be taken with a grain of salt:
 Box is generally less effective than Disjuncts;
due to the incremental nature of our experiments,
it did not attempt as many of the ``harder'' problems as Disjuncts did.}
However, Box is less effective than Disjuncts as the depths increase;
this is expected, as the loss of precision
with more operations is more severe
for Box.

\paragraph{Size of Dataset and Number of Features}
We measure whether the size of the dataset (which in our benchmarks is 
quite correlated with the number of features) affects the performance.
Consider the case of verifying a decision-tree learner of depth 3 using the disjunctive domain
and a perturbed set where
$0.5\%$ of the points\footnote{We round to the closest $n$ for which
the tool can verify at least one instance} are removed from the dataset (similar trends are observed
when varying these parameters).
The average running time of \name is 
0.1s  for \iris, 
0.2s for \mammography, 
26s for \wdbc, and 
32s for \mnistbin. 
For \mnistreal, 100\% of the benchmarks TO at $0.05\%$ poisoning. 
As expected, the size of the dataset and the number of features have an effect on the verification time.
However, it is hard to exactly quantify this effect, given how differently each dataset behaves;
an obvious comparison we can make is the difference between \mnistbin and \mnistreal. These two datasets
have identical sizes, but the former uses binary features and the latter uses real features.
As we can see, handling real features results in a massive slowdown and in proving fewer instances robust. This is not surprising
since real features can result in more predicates, which affect both running time and the discrimination
power of individual nodes in the decision tree.

\input{data/full_mnist_bool}

\paragraph{Depth of the Tree}
Consider the case of verifying a decision-tree learner for \mnistbin using the disjunctive domain,
and a perturbed set where
up to 64 of the points have been added maliciously to the dataset (similar trends are observed
when varying these parameters and for other datasets).
The average running time of \name is 
0.3s  at depth 1, 
0.5s at depth 2, 
32s  at depth 3, and 
933s  at depth 4. 
As expected, the depth of the tree is an important factor in the performance of the disjunctive domain,
as each abstract operation expands the set of disjuncts.

We summarize the results presented in this section and
answer \textbf{RQ2}:
\emph{in general, 
the disjunctive domain is more precise but slower than the non-disjunctive domain, and
the depth of the learned trees and the number of poisoned elements in the dataset are the greatest factors
affecting performance.}


%% file: benchmarktable.tex
\begin{table*}
\caption{Detailed metrics for the benchmark datasets considered in our evaluation. * Test set accuracy for MNIST is computed on the full 2,163 instances; robustness experiments are performed on 100 randomly chosen test set elements.}
\small
\centering

\newcommand{\irisclasses}{\{\text{Setosa}, \text{Versicolour}, \text{Virginica}\}}
\newcommand{\cancerclasses}{\{\text{benign}, \text{malignant}\}}
\newcommand{\mnistclasses}{\{\text{one}, \text{seven}\}}
\newcommand{\fp}[1]{\nprounddigits{1}\numprint{#1}}
\newcommand{\tworow}[1]{\multirow{2}{*}[-0.5\dimexpr \aboverulesep + \belowrulesep + \cmidrulewidth]{#1}}

\begin{tabular}{lrrllrrrr}
\toprule
\tworow{Data Set} & \multicolumn{2}{c}{Size} & \tworow{Features $\tdom$} & \tworow{Classes $\classes$} & \multicolumn{4}{c}{DT Test-Set Accuracy (\%)} \\
\cmidrule(lr){2-3} \cmidrule(lr){6-9}
 & Training & Test & & & Depth 1 & 2 & 3 & 4 \\
\midrule
\iris & 120 & 30 & $\mathds{R}^4$ & $\irisclasses$ & \fp{20} & \fp{90} & \fp{90} & \fp{90} \\
\mammography & 664 & 166 & $\mathds{R}^5$ & $\cancerclasses$ & \fp{80.72289} & \fp{83.13253} & \fp{81.92771} & \fp{80.72289} \\
\wdbc & 456 & 113 & $\mathds{R}^{30}$ & $\cancerclasses$ & \fp{91.15044} & \fp{92.03539} & \fp{92.92035} & \fp{94.69026} \\
\mnistbin & \numprint{13007} & 100* & $\{0,1\}^{784}$ & $\mnistclasses$ & \fp{95.70041} & \fp{97.41100} & \fp{97.82709} & \fp{98.28941} \\
\mnistreal & \numprint{13007} & 100* & $\mathds{R}^{784}$ &  $\mnistclasses$ & \fp{95.60795} & \fp{97.59593} & \fp{98.33564} & \fp{98.70550} \\
\bottomrule
\end{tabular}

\label{ta:benchmarks}
\end{table*}

%% file: data/firstfigure.tex
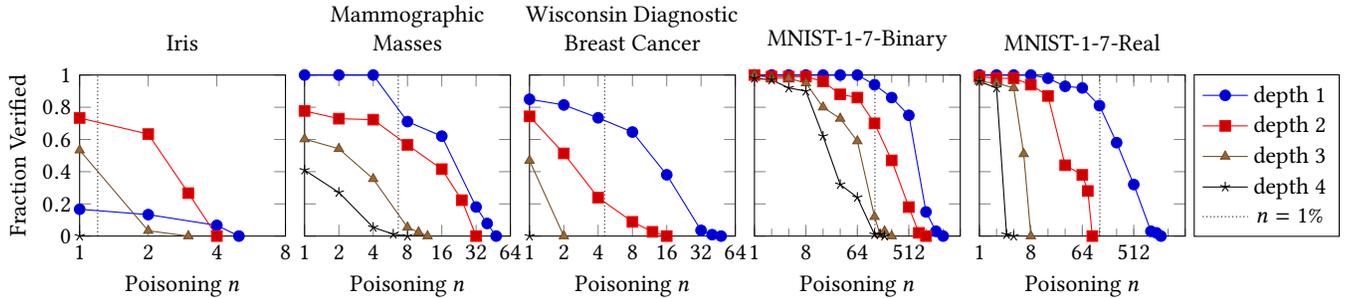
\begin{figure*}[t]
\centering
\small
\pgfplotsset{filter discard warning=false}
\pgfplotscreateplotcyclelist{whatever}{%
    blue,every mark/.append style={fill=blue!80!black},mark=*\\%
    red,every mark/.append style={fill=red!80!black},mark=square*\\%
    brown!60!black,every mark/.append style={fill=brown!80!black},mark=triangle*\\%
    black,mark=star\\%
    blue,every mark/.append style={fill=blue!80!black},mark=diamond*\\%
    red,densely dashed,every mark/.append style={solid,fill=red!80!black},mark=*\\%
    brown!60!black,densely dashed,every mark/.append style={solid,fill=brown!80!black},mark=square*\\%
    black,densely dashed,every mark/.append style={solid,fill=gray},mark=triangle*\\%
    blue,densely dashed,mark=star,every mark/.append style=solid\\%
    red,densely dashed,every mark/.append style={solid,fill=red!80!black},mark=diamond*\\%
    }
\begin{tikzpicture}
    \begin{groupplot}[
            group style={
                group name=pvgroup,
                group size=5 by 1,
                horizontal sep=.1in,
                ylabels at=edge left,
                yticklabels at=edge left,
            },
            width=1.7in,
            xlabel near ticks,
            ylabel near ticks,
            xlabel=Poisoning $n$,
            ylabel=Fraction Verified,
            xmode=log,
            xtick={1,2,4,8,16,32,64,128,256,512,1024,2048},
            log ticks with fixed point,
            xmin=1,
            ymin=0,
            ymax=1,
            cycle list name=whatever,
        ]

        \nextgroupplot[
            title=\iris,
            xmax=8,
        ]
        \foreach \n in {1, 2, 3, 4}{
            \addplot table [x=num_dropout, y=percent_verified, col sep=comma]
                {data/proven_vs_poisoned/iris_d\n.csv};
        }
        \addplot[densely dotted] coordinates {(1.2, 0) (1.2, 1)};

        \nextgroupplot[
            title={\parbox{1in}{\centering \mammography}},
            xmax=64,
        ]
        \foreach \n in {1, 2, 3, 4}{
            \addplot table [x=num_dropout, y=percent_verified, col sep=comma]
                {data/proven_vs_poisoned/mammography_d\n.csv};
        }
        \addplot[densely dotted] coordinates {(6.64, 0) (6.64, 1)};

        \nextgroupplot[
            title={\parbox{1.3in}{\centering \wdbc}},
            xmax=64,
        ]
        \foreach \n in {1, 2, 3, 4}{
            \addplot table [x=num_dropout, y=percent_verified, col sep=comma]
                {data/proven_vs_poisoned/wdbc_d\n.csv};
        }
        \addplot[densely dotted] coordinates {(4.56, 0) (4.56, 1)};

        \nextgroupplot[
            title=\mnistbin,
            xmax=4096,
            xtick={1,8,64,512},
            minor xtick={1,2,4,8,16,32,64,128,256,512,1024,2048,4096},
        ]
        \foreach \n in {1, 2, 3, 4}{
            \addplot table [x=num_dropout, y=percent_verified, col sep=comma]
                {data/proven_vs_poisoned/mnist_simple_1_7_d\n.csv};
        }
        \addplot[densely dotted] coordinates {(130, 0) (130, 1)};

        \nextgroupplot[
            title=\mnistreal,
            xmax=4096,
            xtick={1,8,64,512},
            minor xtick={1,2,4,8,16,32,64,128,256,512,1024,2048,4096},
            legend to name={singlelegend},
        ]
        \foreach \n in {1, 2, 3, 4}{
            \addplot table [x=num_dropout, y=percent_verified, col sep=comma]
                {data/proven_vs_poisoned/mnist_1_7_d\n.csv};
        }
        \addplot[densely dotted] coordinates {(130, 0) (130, 1)};

        \addlegendentry{depth 1}
        \addlegendentry{depth 2}
        \addlegendentry{depth 3}
        \addlegendentry{depth 4}
        \addlegendentry{$n = 1\%$}
    \end{groupplot}
    \node (P1) at (pvgroup c5r1.north east) {};
    \node (P2) at (pvgroup c5r1.south east) {};
    \path (P1) -- node[right]{\pgfplotslegendfromname{singlelegend}} (P2);
\end{tikzpicture}
\caption{Fraction of test instances proven robust versus poisoning parameter $n$ (log scale).
The dotted line is a visual aid, indicating $n$ is 1\% of the training set size.}
\label{fig:mainresults}
\end{figure*}

%% file: data/full_mnist_bool.tex
\begin{figure*}
\centering
\small
\pgfplotsset{filter discard warning=false}
\begin{tikzpicture}
    \begin{groupplot}[
            group style={
                group size=4 by 3,
                horizontal sep=.18in,
                vertical sep=.18in,
                ylabels at=edge left,
                yticklabels at=edge left,
                xlabels at=edge bottom,
                xticklabels at=edge bottom
            },
            width=2in,
            height=1.4in,
            xlabel near ticks,
            ylabel near ticks,
            xlabel=Poisoning $n$,
            xmode=log,
            xtick={1,8,64,512},
            minor xtick={1,2,4,8,16,32,64,128,256,512,1024,2048,4096},
            log ticks with fixed point,
            xmin=1,
            xmax=2048,
        ]

        \nextgroupplot[
            title=Depth 1,
            ylabel=\# Verified,
            ymin=0,
            ymax=100,
        ]
        \addplot table [x=num_dropout, y=num_verified, col sep=comma]{data/exhaustive/mnist_simple_1_7_d1_box.csv};
        \addplot table [x=num_dropout, y=num_verified, col sep=comma]{data/exhaustive/mnist_simple_1_7_d1_disjuncts.csv};

        \nextgroupplot[
            title=Depth 2,
            ymin=0,
            ymax=100,
        ]
        \addplot table [x=num_dropout, y=num_verified, col sep=comma]{data/exhaustive/mnist_simple_1_7_d2_box.csv};
        \addplot table [x=num_dropout, y=num_verified, col sep=comma]{data/exhaustive/mnist_simple_1_7_d2_disjuncts.csv};

        \nextgroupplot[
            title=Depth 3,
            ymin=0,
            ymax=100,
        ]
        \addplot table [x=num_dropout, y=num_verified, col sep=comma]{data/exhaustive/mnist_simple_1_7_d3_box.csv};
        \addplot table [x=num_dropout, y=num_verified, col sep=comma]{data/exhaustive/mnist_simple_1_7_d3_disjuncts.csv};

        \nextgroupplot[
            title=Depth 4,
            ymin=0,
            ymax=100,
        ]
        \addplot table [x=num_dropout, y=num_verified, col sep=comma]{data/exhaustive/mnist_simple_1_7_d4_box.csv};
        \addplot table [x=num_dropout, y=num_verified, col sep=comma]{data/exhaustive/mnist_simple_1_7_d4_disjuncts.csv};
        \legend{Box, Disjuncts};

        \nextgroupplot[
            ylabel=Average Time (s),
            ymode=log,
            ymin=.1,
            ymax=10000,
            ytick={.1,1,10,100,1000,10000},
        ]
        \addplot table [x=num_dropout, y=avg_time, col sep=comma]{data/exhaustive/mnist_simple_1_7_d1_box.csv};
        \addplot table [x=num_dropout, y=avg_time, col sep=comma]{data/exhaustive/mnist_simple_1_7_d1_disjuncts.csv};
        \addplot[dotted, domain=1:4096] {3600};
        \legend{ , , Timeout};

        \nextgroupplot[
            ymode=log,
            ymin=.1,
            ymax=10000,
            ytick={.1,1,10,100,1000,10000},
        ]
        \addplot table [x=num_dropout, y=avg_time, col sep=comma]{data/exhaustive/mnist_simple_1_7_d2_box.csv};
        \addplot table [x=num_dropout, y=avg_time, col sep=comma]{data/exhaustive/mnist_simple_1_7_d2_disjuncts.csv};
        \addplot[dotted, domain=1:4096] {3600};

        \nextgroupplot[
            ymode=log,
            ymin=.1,
            ymax=10000,
            ytick={.1,1,10,100,1000,10000},
        ]
        \addplot table [x=num_dropout, y=avg_time, col sep=comma]{data/exhaustive/mnist_simple_1_7_d3_box.csv};
        \addplot table [x=num_dropout, y=avg_time, col sep=comma]{data/exhaustive/mnist_simple_1_7_d3_disjuncts.csv};
        \addplot[dotted, domain=1:4096] {3600};

        \nextgroupplot[
            ymode=log,
            ymin=.1,
            ymax=10000,
            ytick={.1,1,10,100,1000,10000},
        ]
        \addplot table [x=num_dropout, y=avg_time, col sep=comma]{data/exhaustive/mnist_simple_1_7_d4_box.csv};
        \addplot table [x=num_dropout, y=avg_time, col sep=comma]{data/exhaustive/mnist_simple_1_7_d4_disjuncts.csv};
        \addplot[dotted, domain=1:4096] {3600};

        \nextgroupplot[
            ylabel=\parbox{1in}{\centering Average Max Memory (MB)},
            ymode=log,
            ymax=200000,
        ]
        \addplot table [x=num_dropout, y=avg_max_memory, col sep=comma]{data/exhaustive/mnist_simple_1_7_d1_box.csv};
        \addplot table [x=num_dropout, y=avg_max_memory, col sep=comma]{data/exhaustive/mnist_simple_1_7_d1_disjuncts.csv};
        \addplot[dotted, domain=1:4096] {140000};
        \legend{ , , OOM};

        \nextgroupplot[
            ymode=log,
            ymax=200000,
        ]
        \addplot table [x=num_dropout, y=avg_max_memory, col sep=comma]{data/exhaustive/mnist_simple_1_7_d2_box.csv};
        \addplot table [x=num_dropout, y=avg_max_memory, col sep=comma]{data/exhaustive/mnist_simple_1_7_d2_disjuncts.csv};
        \addplot[dotted, domain=1:4096] {140000};

        \nextgroupplot[
            ymode=log,
            ymax=200000,
        ]
        \addplot table [x=num_dropout, y=avg_max_memory, col sep=comma]{data/exhaustive/mnist_simple_1_7_d3_box.csv};
        \addplot table [x=num_dropout, y=avg_max_memory, col sep=comma]{data/exhaustive/mnist_simple_1_7_d3_disjuncts.csv};
        \addplot[dotted, domain=1:4096] {140000};

        \nextgroupplot[
            ymode=log,
            ymax=200000,
        ]
        \addplot table [x=num_dropout, y=avg_max_memory, col sep=comma]{data/exhaustive/mnist_simple_1_7_d4_box.csv};
        \addplot table [x=num_dropout, y=avg_max_memory, col sep=comma]{data/exhaustive/mnist_simple_1_7_d4_disjuncts.csv};
        \addplot[dotted, domain=1:4096] {140000};
    \end{groupplot}
\end{tikzpicture}
\caption{Efficacy, performance, and memory usage for \mnistbin}
\label{fig:mnistbin}
\end{figure*}
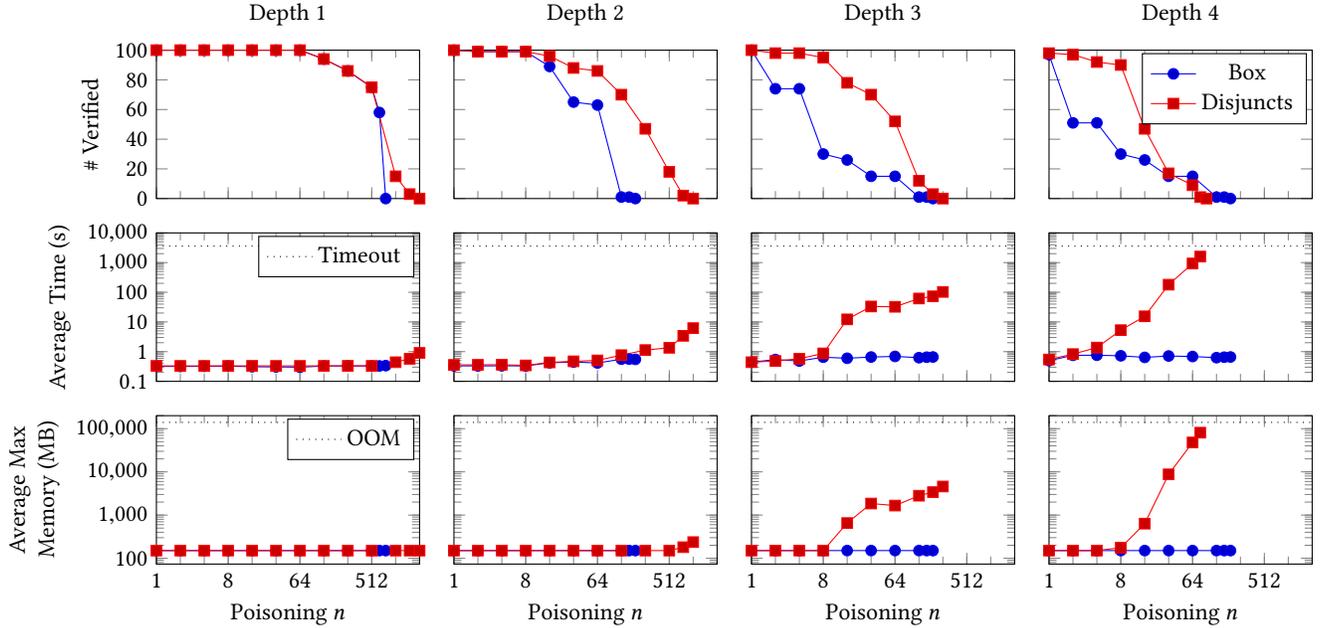

%% file: relatedwork.tex
\section{Related Work}\label{sec:relatedwork}

\paragraph{Instability in Decision Trees}
Decision-tree learning has a long and storied history. 
 A particular
thread of work that is relevant to ours is the analysis of decision-tree \emph{instability}~\cite{dwyer07,turney95,li02,perez05}.
These works show that decision-tree learning algorithms are in general susceptible to small data-poisoning attacks---%
although they do not phrase it in those terms.
For the most part, the works are motivated from the perspective
that a decision tree represents a set of ``rules,''
and they are concerned with conditions under which those rules will not change
(either by quantifying forms of invariance or providing novel learning algorithms).
Our work is different in that it \textit{proves} that no poisoning attack exists
on a formalization of very basic decision-tree learning,
and we can often precisely allow for the ``rules'' to change
so long as the ultimate classification does not.

\paragraph{Data Poisoning}
Data-poisoning robustness has been studied extensively from an attacker 
perspective~\cite{Biggio12,Xiao12,Xiao15,Newell14,Mei15}. 
This body
of work has demonstrated  attacks that can degrade classifier accuracy, sometimes
dramatically. 
These works phrase the problem of identifying a poisoned set  as a constraint optimization problem.
To make the problem tractable, they typically focus on support vector machines (SVMs)
and forms of regression for which existing optimization techniques are readily available.
Our approach differs from these works in multiple ways:
\rone Our work focuses on \textit{decision trees}.
The greedy, recursive nature of decision-tree learning is fundamentally different
from the optimization problem solved in learning SVMs.
\rtwo While our technique is general, in this paper we consider a poisoning model in which training elements have been added~\cite{xiao2015feature,chen2017targeted}. 
Some  works instead focuses on a model in which elements of the training set can be modified~\cite{alfeld2016data}.
\rthree Final and most important, our work \textit{proves} that no poisoning attack exists using 
abstract interpretation, while existing techniques largely provide search techniques for finding poisoned training sets. 

Recently, techniques have been proposed to modify the training processes of machine learning models
to make them robust to various
data-poisoning attacks (while remaining computationally efficient).
These techniques~\cite{LaishramP16, steinhardt2017certified, Diakonikolas19, diakonikolas19sever}
are often based on robust estimation, e.g.\ outlier removal;
see~\cite{diakonikolas19survey} for a survey.
In general, these approaches
provide limited probabilistic guarantees about certain kinds of attacks;
the works are orthogonal to ours, though they raise an interesting question for future work:
Can one verify that, on a given training set, these models actually make
the training process resistant to data poisoning?

\paragraph{Abstract Interpretation for Robustness}
Abstract interpretation~\cite{Cousot77}
is one of the most popular models for static program analysis.
Our work is inspired by that of \citet{Gehr18},
where  abstract interpretation is used to prove
input-robustness for neural networks.
(Recently, Ranzato and Zanella
have done similar work for decision tree ensembles~\cite{ranzato20}.)
Many papers have followed improving on this problem~\cite{anderson2019optimization,singh2019abstract}.
The main difference between these works and ours is that
we tackle the problem of verifying training-time robustness, while existing works focus on test-time robustness.
The former problem requires abstracting sets of training sets,
while the latter only requires abstracting sets of individual inputs. In particular, 
Gehr et al. rely on well-known abstract domains---e.g., intervals and zonotopes---to represent
sets of real vectors, while our work presents entirely new abstract domains
for reasoning about sets of training sets.
To our knowledge, our work is the first that even tries to tackle
the problem of verifying data-poisoning robustness.

Other works have focused on \emph{provable training} of neural networks
to exhibit test-time robustness by construction~\cite{wong18, mirman18}:
this is done by using abstract interpretation to over-approximate the worst-case loss
formed by any adversarial perturbation to any element in the training set.
One can think of these techniques as performing a form of symbolic training,
which is conceptually similar to our core idea.
Note, however, two important distinctions:
\rone These works address the problem of adversarial changes to test inputs,
while we address adversarial changes to the training set;
\rtwo These works construct a different, robust model,
while we verify a property of an unchanged model (or rather, the learner).

%% file: conclusion.tex
\section{Conclusion}\label{sec:conclusion}

We presented \name, the first tool that can  verify  
 data-poisoning robustness for decision tree learners,
 where an attacker may have contributed malicious training items.
\name is based on abstract interpretation and introduces
a new abstract domain for representing sets of training sets.
We showed that \name can verify robustness for real-world datasets
in cases  where an enumeration approach would be intractable.
To our knowledge, this paper is the first to verify data-poisoning robustness
for any kind of machine learning model. 
A natural future direction is to extend our ideas to 
 neural networks, where the learning algorithm is stochastic.

\paragraph{Acknowledgements}
We would like to thank Jerry Zhu, Ilias Diakonikolas, and Paris Koutris
for their feedback,
as well as our shepherd, Martin Vechev.
This material is based upon work supported by the National Science Foundation
under grant numbers 1652140, 1704117, 1750965, and 1918211.

%% file: app-proofs.tex
\section{Proofs of Soundness for $\atlearner$} \label{app:proofs}

\begin{proof}[Proof of Proposition~\ref{prop:join}]
Suppose $\train \in \gamma\left(\drop{\train_1}{n_1}\right)\cup \gamma(\drop{\train_2}{n_2})$.
Without loss of generality, assume $T \in \gamma(\drop{\train_1}{n_1})$.
Certainly ${\train \subseteq \train_1 \cup \train_2}$;
furthermore, observe that $\train$ can be formed by first removing
$|\train_2 \setminus \train_1|$-many elements from $\train_1 \cup \train_2$
to recover $\train_1$ and then removing $\leq n_1$ remaining elements:
this gives us ${|(\train_1 \cup \train_2) \setminus \train|}
= {|(\train_2 \setminus \train_1) \cup (\train_1 \setminus \train)|}
\leq {|\train_2 \setminus \train_1| + |\train_1 \setminus \train|}
\leq {|\train_2 \setminus \train_1| + n_1}$.
This matches the definition of $\sqcup$.
\end{proof}

\begin{proof}[Proof of Proposition~\ref{prop:arestr}]
    Let $\drop{S}{m} = \gamma(\arestr{\drop{\train}{n}}{\varphi})$
    (so we will show $\restr{\train'}{\varphi} \in \gamma(\drop{S}{m})$).
    Since $\train' \subseteq \train$,
    we have that
    $\restr{\train'}{\varphi} \subseteq \{(\point,\class) \in \train : \point \models \varphi\} = S$,
    thus $\restr{\train'}{\varphi} \subseteq S$.
    Additionally, $S \setminus \restr{\train'}{\varphi}
    = S \setminus \train'
    \subseteq \train \setminus \train'$;
    certainly, then, $|S \setminus \restr{\train'}{\varphi}| \leq |\train \setminus \train'|$ and thus $ \leq n$.
    Recall there are two cases, $m = n$ or $m = |S|$:
    for the latter, we have trivially that $|S \setminus \restr{\train'}{\varphi}| \leq |S| \leq m$.
    Therefore $|S \setminus \restr{\train'}{\varphi}| \leq \min(n, m)$.
\end{proof}

\begin{proof}[Proof of Proposition~\ref{prop:asummary}]
    When $n < |\train|$, this follows from the soundness of interval arithmetic.
    In our corner case for $n = |\train|$,
    the concrete $\summary$ is undefined behavior,
    so we encompass every well-formed categorical probability distribution
    by allowing each component to take any $[0,1]$ value.
\end{proof}

\begin{proof}[Proof of Proposition~\ref{prop:afilter}]
    This is an immediate consequence
    of the soundness of $\arestr{\drop{\train}{n}}{\varphi'}$ and the soundness of the join.
\end{proof}

\begin{proof}[Proof of Lemma~\ref{lem:abestsplit}]
First, observe that $\ascore$ is sound
since it involves composing sound interval arithmetic
with other sound operations.
Let $\varphi' = \bestsplit(\train')$.
\begin{itemize}
    \item Case $\varphi' = \nullpred$:
        By definition of $\bestsplit$, this occurs when $\train'$ is such that
        every predicate $\varphi \in \Phi$ results in trivial splits.
        Soundness of $\arestr{}{}$ then gives us that for each $\varphi \in \Phi$,
        either $\emptyset \in \gamma(\arestr{\drop{\train}{n}}{\varphi})$
        or $\emptyset \in \gamma(\arestr{\drop{\train}{n}}{\neg\varphi})$,
        thus $\Phi_\forall = \emptyset$.
        We return using the then-branch, which explicitly includes $\nullpred$.
    \item Case $\varphi' \neq \nullpred$:
        By definition of $\bestsplit$,
        \rone $\varphi'$ non-trivially splits $\train'$, and furthermore,
        \rtwo for all other $\psi$ that non-trivially split $\train'$,
        we have that $\score(\train', \varphi') \leq \score(\train', \psi)$.
        \rone gives us that $\varphi' \in \Phi_\exists$,
        therefore if $\Phi_\forall = \emptyset$, then $\varphi'$ is included in the return value.

        Otherwise, when $\Phi_\forall \neq \emptyset$,
        we return using the then-branch.
        Let $\psi^*$ minimize $\lub_{\Phi_\forall}$:
        since $\psi^* \in \Phi_\forall$, \rtwo gives us that
        $\score(\train', \varphi') \leq \score(\train', \psi^*)$,
        and thus we have $\lb(\ascore(\drop{\train}{n}, \varphi'))
        \leq \ub(\ascore(\drop{\train}{n}, \psi^*))
        = \lub_{\Phi_\forall}$.
        Therefore $\varphi'$ is included in the return.
\end{itemize}
\end{proof}

\begin{proof}[Proof of Theorem~\ref{thm:main}]
$\atlearner$ applies a sequence of operations;
throughout the section, we state the soundness of each of these operations.
The soundness of $\atlearner$ follows from taking their composition.
\end{proof}

%% file: app-real.tex
\section{Real-Valued Features}
\label{app:real-valued-features}

In this appendix, we provide a complete exposition of the technique used to handle
real-valued features. We repeat some of the arguments given in the main paper 
to make this appendix self-contained.
For real-valued features, there are infinitely many possible predicates of the form
$\lambda \point_i \ldotp \point_i \leq \tau$ (where $\tau \in \mathds{R}$),
and the learner $\tlearnerr$ chooses a finite set of possible $\tau$ values dynamically,
based on the values that occur in $\train$.
Throughout this section, we use the subscript $\reals$ to denote the real-valued
versions of existing operations.

\subsection{From $\tlearner$ to $\tlearnerr$}
The new learner $\tlearnerr$ is almost identical to $\tlearner$.
However, to formalize the aforementioned operation in the concrete semantics  
of $\tlearnerr$, we make a single modification in $\bestsplitr$, which
 maintains the original definition of $\bestsplit$, but it first computes 
a finite set of predicates $\Phi_\reals$
used in the remainder of its computation:
consider all of the values appearing in $\train$ for the $i$th feature in $\tdom$,
sorted in ascending order.
For each pair of adjacent values $(a, b)$ (i.e., such that there exists no $c$ in $\train$
such that $a<c<b$),
we include in $\Phi_\reals$ the predicate $\varphi=\lambda \point_i \ldotp \point_i \leq \frac{a + b}{2}$.

\begin{example}
In our running example from Figure~\ref{fig:example},
we have training set elements in $\train_\emph{bw}$
whose features take the numeric values $\{0,1,2,3,4,7,\ldots,14\}$.
$\bestsplitr(\train_\emph{bw})$ would thus enumerate over the predicates
$\Phi_\reals = \{\lambda \point \ldotp \point \leq \tau \mid
\tau \in \{\frac{1}{2},\frac{3}{2},\frac{5}{2},\frac{7}{2},\frac{11}{2},\frac{15}{2},\ldots,\frac{27}{2}\} \}$.
\end{example}

\subsection{From $\atlearner$ to $\atlearnerr$}
The formalization for the abstract case is more involved than the concrete case:
\rone we will similarly modify $\abestsplit$, but also
\rtwo we will change our abstract domain over predicates.
This change to the predicate domain means we will have to make
largely superficial adjustments to the many of the other operations, as well.

$\abestsplitr$, the real-valued version of $\abestsplit$, is responsible for selecting a finite set of predicates
that it will consider in its computation.
This motivates the second point, which we will discuss first:
because we don't know \emph{which} values could be missing from $\drop{\train}{n}$,
we might naively consider a value of $\tau$ for every such possible combination of missing values.
If we did so, and if $\bestsplitr{\train}$ considers $m$ predicates,
then $\abestsplitr(\drop{\train}{n})$ might consider up to $\approx mn$ predicates.
For efficiency, we will instead create a finite set of $m$-many \emph{symbolic predicates}
that overapproximates these possibilities.

\begin{definition}[Symbolic Real-Valued Predicate]
For a real-valued feature at $\point_i$, a \emph{symbolic predicate} $\rpred$ over $\point_i$ takes the form
$\lambda \point_i \ldotp \point_i \leq [a, b)$
for real values $a$ and $b$.
The semantics of a symbolic predicate is three-valued, which we denote as follows:
\[
\rpred(\point) \coloneqq 
\begin{cases}
\mathit{true} & \point_i \leq a \\
\mathit{maybe} & a < \point_i < b \\
\mathit{false} & b \leq \point_i \\
\end{cases}
\]
Each symbolic predicate $\rpred=\lambda \point_i \ldotp \point_i \leq [a, b)$ has the concretization
$\gamma(\rpred) \coloneqq \{\lambda \point_i \ldotp \point_i \leq \tau \mid \tau \in [a, b)\}$.
\end{definition}

Without loss of generality, we focus on the single disjunct case.
We previously stated $\atlearner$ operates over a state $(\drop{\train}{n}, \Psi)$
where $\Psi$ is some finite set of predicates;
now we let the state 
of $\atlearnerr$ be $(\drop{\train}{n}, \Psi^\#)$
where $\Psi^\#$ is represented as a finite set of \emph{symbolic predicates}.
The join remains a simple set union
$\Psi^\#_1 \sqcup \Psi^\#_2 \coloneqq \Psi^\#_1 \cup \Psi^\#_2$.
The concretization captures an infinite set of (non-symbolic) predicates:
$\gamma(\Psi^\#) \coloneqq \bigcup_{\rpred \in \Psi^\#} \gamma(\rpred)$.

\paragraph{Symbolic Predicates in Auxiliary Operators}
The definition of $\arestr{\drop{\train}{n}}{\rpred}$ changes slightly,
since we must include the \emph{maybe} case to be sound.
This complicates the computation of the poisoning amount $n$,
since we have two sources of uncertainty:
we must account for elements that could be missing because either
\rone they are missing from some particular $\train' \in \gamma(\drop{\train}{n})$, or
\rtwo $\rpred$ evaluates to \emph{maybe}.
Fortunately, we will be able to succinctly encompass these possibilities
using our existing lower-level operations.

We define $\arestr{\drop{\train}{n}}{\rpred}$ as follows:
suppose $\rpred$ is of the form $\lambda \point_i \ldotp \point_i \leq [a, b)$,
let $\varphi_a = \lambda \point_i \ldotp \point_i \leq a$
and let $\varphi_b = \lambda \point_i \ldotp \point_i < b$.
Because $\varphi_a$ and $\varphi_b$ are concrete predicates, we can use 
Equation~\ref{eq:abstract-semantics-dropn} (Section~\ref{sec:auxops})
to compute their abstract semantics.
Then
\[
\arestr{\drop{\train}{n}}{\rpred}:=\arestr{\drop{\train}{n}}{\varphi_a} \sqcup \arestr{\drop{\train}{n}}{\varphi_b}
.
\]

\begin{proposition}
Let $\train' \in \gamma(\drop{\train}{n})$ and $\varphi' \in \gamma(\rpred)$.
Then,
\[
\restr{\train'}{\varphi'} \in \gamma(\arestr{\drop{\train}{n}}{\rpred})
.
\]
\end{proposition}
\begin{proof}
    Denote $\drop{S_a}{m_a} = \arestr{\drop{\train}{n}}{\varphi_a}$
    and $\drop{S_b}{m_b} = \arestr{\drop{\train}{n}}{\varphi_b}$.
    We will show that $\restr{\train'}{\varphi'} \in \gamma(\drop{S_a}{m_a} \sqcup \drop{S_b}{m_b})$.
    This join has nice structure since $a \leq b$ and thus $S_a \subseteq S_b$.

    We break into two cases, since $m_b = \min(n, |S_b|)$.
    First, suppose $m_b = |S_b|$.
    Then $\drop{S_a}{m_a} \sqcup \drop{S_b}{m_b} = \drop{S_b}{|S_b|}$,
    where $\gamma(\drop{S_b}{|S_b|}) = \mathcal{P}(\restr{\train}{\varphi_b})$.
    Because $\varphi' \in \gamma(\rpred)$,
    we then have $\restr{\train'}{\varphi'} \subseteq \restr{\train'}{\varphi_b}
    \in \mathcal{P}(\restr{\train}{\varphi_b})$, and we are done.

    Otherwise, we have $m_b = n$.
    Here, $\drop{S_a}{m_a} \sqcup \drop{S_b}{m_b} = \drop{S_b}{|S_b \setminus S_a| + n}$
    (unless $m_a = |S_a|$, in which case we immediately collapse to the previous case).
    Again, because $\varphi' \in \gamma(\rpred)$, we immediately have that
    $\restr{\train'}{\varphi'} \subseteq S_b$;
    it remains to show that $|S_b \setminus \restr{\train'}{\varphi'}| \leq
    |S_b \setminus S_a| + n$.
    Observe that \rone $\restr{\train}{\varphi'} \setminus \restr{\train'}{\varphi'}
    \subseteq \train \setminus \train'$ and thus
    $|\restr{\train}{\varphi'} \setminus \restr{\train'}{\varphi'}| \leq n$,
    and \rtwo since $S_a \subseteq S_b$ and $\varphi' \in \gamma(\rpred)$,
    we have that $|S_b \setminus \restr{\train}{\varphi'}| \leq |S_b \setminus S_a|$.
    Combined, we know that
    $|S_b \setminus \restr{\train'}{\varphi'}| =
    |S_b \setminus \restr{\train}{\varphi'}| +
    |\restr{\train}{\varphi'} \setminus \restr{\train'}{\varphi'}|
    \leq |S_b \setminus S_a| + n$.
\end{proof}

\paragraph{Symbolic Predicates in $\afilterr$}
In the original definition of $\afilter$
we separated $\Psi$ into two sets $\Psi_{\point}$ and $\Psi_{\neg\point}$
because exclusively either $\varphi \models \point$ or $\neg\varphi \models \point$.
In this new three-valued symbolic predicate case,
we appropriately over-approximate the $\rpred(\point) = \mathit{maybe}$ possibility.

Let us denote the following:
\begin{align*}
\Psi^\#_\point &= \{\rpred \in \Psi^\# \mid \rpred(\point) \in \{\mathit{true}, \mathit{maybe}\}\} \\
\Psi^\#_{\neg\point} &= \{\rpred \in \Psi^\# \mid \rpred(\point) \in \{\mathit{maybe}, \mathit{false}\}\}
\end{align*}
and define
\[
\afilterr(\drop{\train}{n}, \Psi^\#, \point) \coloneqq
\left(\bigsqcup_{\rpred \in \Psi^\#_\point} \arestr{\drop{\train}{n}}{\rpred}\right)
\sqcup \left(\bigsqcup_{\rpred \in \Psi^\#_{\neg \point}} \arestr{\drop{\train}{n}}{\neg\rpred}\right)
\]

\begin{proposition}
Let $\train' \in \gamma(\drop{\train}{n})$ and let $\varphi' \in \gamma(\Psi^\#)$.
Then,
\[
\filterr(\train', \varphi', \point) \in \gamma\left(\afilterr(\drop{\train}{n}, \Psi^\#, \point)\right)
\]
\end{proposition}
\begin{proof}
    Because $\varphi' \in \gamma(\Psi^\#)$,
    we know there is some $\rpred \in \Psi^\#$
    such that $\varphi' \in \gamma(\rpred)$.
    Once again, soundness now follows from the soundness of the join
    and the soundness of $\arestr{}{\rpred}$.
\end{proof}

\paragraph{Symbolic Predicates in $\abestsplitr$}
Finally, we discuss the treatment of our favorite complicated instruction.
Perhaps surprisingly, very little has to change.
Indeed, $\abestsplitr(\drop{\train}{n})$ begins by creating
a finite set of symbolic predicates $\Phi^\#$
and then proceeds exactly as its original definition.

The construction of $\Phi^\#$ is very simple.
In the concrete case we considered $\lambda \point_i \ldotp \point_i \leq \frac{a+b}{2}$
for all adjacent pairs $(a,b)$ of the feature values that occur in $\train$;
here, we will consider $\lambda \point_i \ldotp \point_i \leq [a, b)$
for all such pairs.

It should be surprising that this computation is sound.
After all, it involves (effectively) computing the $\argmin$ over a set of predicates;
here, we're \emph{over-approximating} that set,
and we ought to be concerned that our over-approximation could include
extraneous predicates whose valuation is so small that they occlude
the feasible minimizing predicates.
Serendipitously, the choice of $\Phi^\#$ prevents this from happening.

\begin{lemma}
Let $\train' \in \gamma(\drop{\train}{n})$.
Then,
\[
\bestsplitr(\train') \in
\gamma\left(\abestsplitr(\drop{\train}{n})\right)
\]
\end{lemma}
\begin{proof}
    We begin by observing two facts about the $\Phi^\#$ constructed in $\abestsplitr$:
    \rone For every $\varphi \in \Phi$ built during the non-symbolic $\bestsplitr(\train')$,
    there exists some $\rpred \in \Phi^\#$ such that $\varphi \in \gamma(\rpred)$.
    \rtwo For every $\rpred \in \Phi^\#$, there exists some
    $\train'' \in \gamma(\drop{\train}{n})$ and $\varphi'' \in \Phi$ from $\bestsplitr(\train'')$
    ($\varphi''$ is not necessarily returned as optimal, just constructed for consideration)
    such that $\varphi'' \in \gamma(\rpred)$.

    Let $\varphi' = \bestsplitr(\train')$.
    The rest of the proof is exactly the same as the soundness proof for $\abestsplit$
    (the previous two observations ensure that $\Phi^\#_\forall$ and $\Phi^\#_\exists$
    preserve the properties necessary for those arguments),
    except for establishing that our target $\varphi'$ is in the concretization of some $\rpred'$
    returned in the then-branch of the definition.
    Take any $\varphi^* \in \gamma(\rpred^*)$ where $\rpred^*$ is the minimizer in $\lub_{\Phi^\#}$.
    Observe the following for any $\rpred' \in \gamma(\Phi^\#)$
    such that $\varphi' \in \gamma(\rpred')$:
    \begin{align*}
    \lb(\ascore(\drop{\train}{n}, \rpred'))
    & \leq \lb(\ascore(\drop{\train}{n}, \varphi')) \\
    & \leq \score(\train', \varphi') \\
    & \leq \score(\train', \varphi^*) \\
    & \leq \ub(\ascore(\drop{\train}{n}, \varphi^*)) \\
    & \leq \ub(\ascore(\drop{\train}{n}, \rpred^*)) \\
    & = \lub_{\Phi^\#}
    \end{align*}
    and thus $\varphi'$ would be included.
\end{proof}

%% file: data/bench.tex
\section{Full Benchmarks}\label{app:bench}
Figures~\ref{fig:iris-details}, \ref{fig:mammography-details}, \ref{fig:wdbc-details}, and \ref{fig:mnistreal-details}
show the detailed performance metrics of the remaining datasets.

\input{data/full_iris}
\input{data/full_mammography}
\input{data/full_wdbc}
\input{data/full_mnist_real}

%% file: data/full_iris.tex
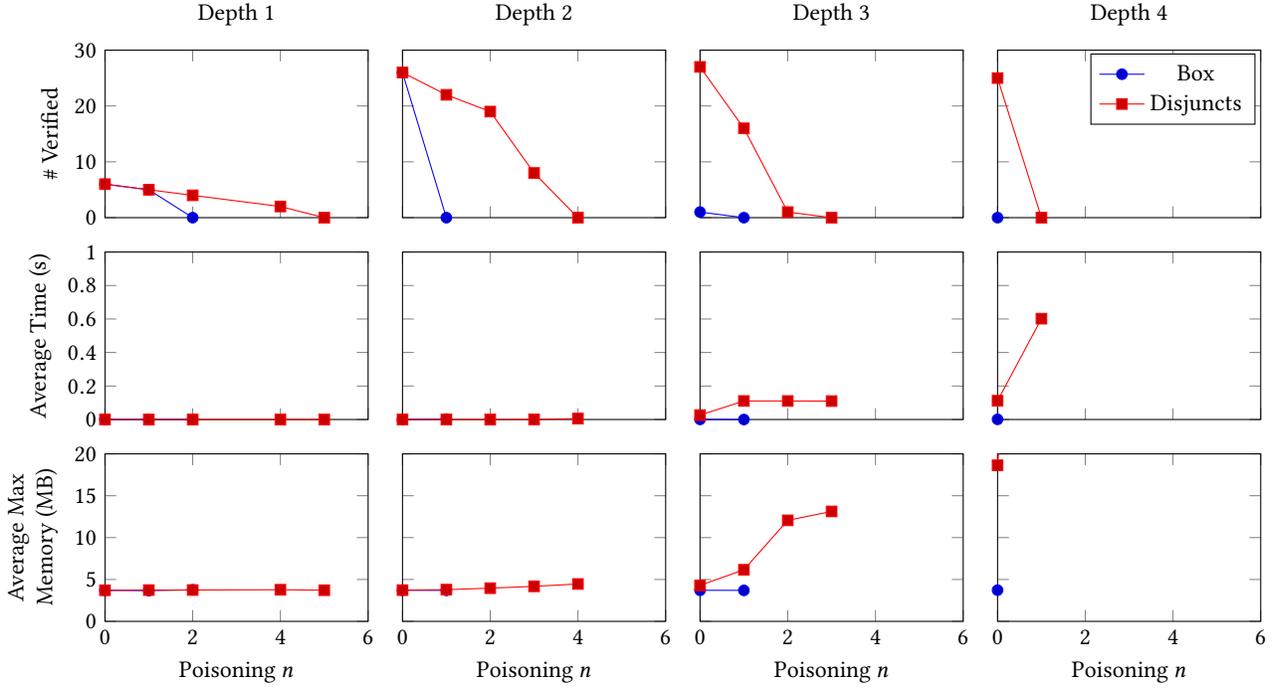
\begin{figure*}
\centering
\small
\pgfplotsset{filter discard warning=false}
\begin{tikzpicture}
    \begin{groupplot}[
            group style={
                group size=4 by 3,
                horizontal sep=.18in,
                vertical sep=.18in,
                ylabels at=edge left,
                yticklabels at=edge left,
                xlabels at=edge bottom,
                xticklabels at=edge bottom
            },
            width=2in,
            height=1.5in,
            xlabel near ticks,
            ylabel near ticks,
            xlabel=Poisoning $n$,
            xmin=0,
            xmax=6,
            ymin=0,
        ]

        \nextgroupplot[
            title=Depth 1,
            ylabel=\# Verified,
            ymax=30,
        ]
        \addplot table [x=num_dropout, y=num_verified, col sep=comma]{data/exhaustive/iris_d1_box.csv};
        \addplot table [x=num_dropout, y=num_verified, col sep=comma]{data/exhaustive/iris_d1_disjuncts.csv};

        \nextgroupplot[
            title=Depth 2,
            ymax=30,
        ]
        \addplot table [x=num_dropout, y=num_verified, col sep=comma]{data/exhaustive/iris_d2_box.csv};
        \addplot table [x=num_dropout, y=num_verified, col sep=comma]{data/exhaustive/iris_d2_disjuncts.csv};

        \nextgroupplot[
            title=Depth 3,
            ymax=30,
        ]
        \addplot table [x=num_dropout, y=num_verified, col sep=comma]{data/exhaustive/iris_d3_box.csv};
        \addplot table [x=num_dropout, y=num_verified, col sep=comma]{data/exhaustive/iris_d3_disjuncts.csv};

        \nextgroupplot[
            title=Depth 4,
            ymax=30,
        ]
        \addplot table [x=num_dropout, y=num_verified, col sep=comma]{data/exhaustive/iris_d4_box.csv};
        \addplot table [x=num_dropout, y=num_verified, col sep=comma]{data/exhaustive/iris_d4_disjuncts.csv};
        \legend{Box, Disjuncts};

        \nextgroupplot[
            ylabel=Average Time (s),
            ymax=1,
        ]
        \addplot table [x=num_dropout, y=avg_time, col sep=comma]{data/exhaustive/iris_d1_box.csv};
        \addplot table [x=num_dropout, y=avg_time, col sep=comma]{data/exhaustive/iris_d1_disjuncts.csv};

        \nextgroupplot[
            ymax=1,
        ]
        \addplot table [x=num_dropout, y=avg_time, col sep=comma]{data/exhaustive/iris_d2_box.csv};
        \addplot table [x=num_dropout, y=avg_time, col sep=comma]{data/exhaustive/iris_d2_disjuncts.csv};

        \nextgroupplot[
            ymax=1,
        ]
        \addplot table [x=num_dropout, y=avg_time, col sep=comma]{data/exhaustive/iris_d3_box.csv};
        \addplot table [x=num_dropout, y=avg_time, col sep=comma]{data/exhaustive/iris_d3_disjuncts.csv};

        \nextgroupplot[
            ymax=1,
        ]
        \addplot table [x=num_dropout, y=avg_time, col sep=comma]{data/exhaustive/iris_d4_box.csv};
        \addplot table [x=num_dropout, y=avg_time, col sep=comma]{data/exhaustive/iris_d4_disjuncts.csv};

        \nextgroupplot[
            ylabel=\parbox{1in}{\centering Average Max Memory (MB)},
            ymax=20,
        ]
        \addplot table [x=num_dropout, y=avg_max_memory, col sep=comma]{data/exhaustive/iris_d1_box.csv};
        \addplot table [x=num_dropout, y=avg_max_memory, col sep=comma]{data/exhaustive/iris_d1_disjuncts.csv};

        \nextgroupplot[
            ymax=20,
        ]
        \addplot table [x=num_dropout, y=avg_max_memory, col sep=comma]{data/exhaustive/iris_d2_box.csv};
        \addplot table [x=num_dropout, y=avg_max_memory, col sep=comma]{data/exhaustive/iris_d2_disjuncts.csv};

        \nextgroupplot[
            ymax=20,
        ]
        \addplot table [x=num_dropout, y=avg_max_memory, col sep=comma]{data/exhaustive/iris_d3_box.csv};
        \addplot table [x=num_dropout, y=avg_max_memory, col sep=comma]{data/exhaustive/iris_d3_disjuncts.csv};

        \nextgroupplot[
            ymax=20,
        ]
        \addplot table [x=num_dropout, y=avg_max_memory, col sep=comma]{data/exhaustive/iris_d4_box.csv};
        \addplot table [x=num_dropout, y=avg_max_memory, col sep=comma]{data/exhaustive/iris_d4_disjuncts.csv};
    \end{groupplot}
\end{tikzpicture}
\caption{\iris (Note: this is the only benchmark for which we do not use log-log plots,
since the numbers are generally small.
\label{fig:iris-details}
}
\end{figure*}

%% file: data/full_mammography.tex
\begin{figure*}
\centering
\small
\pgfplotsset{filter discard warning=false}
\begin{tikzpicture}
    \begin{groupplot}[
            group style={
                group size=4 by 3,
                horizontal sep=.18in,
                vertical sep=.18in,
                ylabels at=edge left,
                yticklabels at=edge left,
                xlabels at=edge bottom,
                xticklabels at=edge bottom
            },
            width=2in,
            height=1.5in,
            xlabel near ticks,
            ylabel near ticks,
            xlabel=Poisoning $n$,
            xmode=log,
            xtick={1,2,4,8,16,32,64,128,256,512,1024,2048},
            log ticks with fixed point,
            xmin=1,
            xmax=64,
        ]

        \nextgroupplot[
            title=Depth 1,
            ylabel=\# Verified,
            ymin=0,
            ymax=166,
        ]
        \addplot table [x=num_dropout, y=num_verified, col sep=comma]{data/exhaustive/mammography_d1_box.csv};
        \addplot table [x=num_dropout, y=num_verified, col sep=comma]{data/exhaustive/mammography_d1_disjuncts.csv};

        \nextgroupplot[
            title=Depth 2,
            ymin=0,
            ymax=166,
        ]
        \addplot table [x=num_dropout, y=num_verified, col sep=comma]{data/exhaustive/mammography_d2_box.csv};
        \addplot table [x=num_dropout, y=num_verified, col sep=comma]{data/exhaustive/mammography_d2_disjuncts.csv};

        \nextgroupplot[
            title=Depth 3,
            ymin=0,
            ymax=166,
        ]
        \addplot table [x=num_dropout, y=num_verified, col sep=comma]{data/exhaustive/mammography_d3_box.csv};
        \addplot table [x=num_dropout, y=num_verified, col sep=comma]{data/exhaustive/mammography_d3_disjuncts.csv};

        \nextgroupplot[
            title=Depth 4,
            ymin=0,
            ymax=166,
        ]
        \addplot table [x=num_dropout, y=num_verified, col sep=comma]{data/exhaustive/mammography_d4_box.csv};
        \addplot table [x=num_dropout, y=num_verified, col sep=comma]{data/exhaustive/mammography_d4_disjuncts.csv};
        \legend{Box, Disjuncts};

        \nextgroupplot[
            ylabel=Average Time (s),
            ymode=log,
            ymin=.00001,
            ymax=10,
            ytick={.00001,.0001,.001,.01,.1,1,10},
            yticklabels={$10^{-5}$,$10^{-4}$,$10^{-3}$,$10^{-2}$,$10^{-1}$,1,10}
        ]
        \addplot table [x=num_dropout, y=avg_time, col sep=comma]{data/exhaustive/mammography_d1_box.csv};
        \addplot table [x=num_dropout, y=avg_time, col sep=comma]{data/exhaustive/mammography_d1_disjuncts.csv};

        \nextgroupplot[
            ymode=log,
            ymin=.00001,
            ymax=10,
            ytick={.00001,.0001,.001,.01,.1,1,10},
        ]
        \addplot table [x=num_dropout, y=avg_time, col sep=comma]{data/exhaustive/mammography_d2_box.csv};
        \addplot table [x=num_dropout, y=avg_time, col sep=comma]{data/exhaustive/mammography_d2_disjuncts.csv};

        \nextgroupplot[
            ymode=log,
            ymin=.00001,
            ymax=10,
            ytick={.00001,.0001,.001,.01,.1,1,10},
        ]
        \addplot table [x=num_dropout, y=avg_time, col sep=comma]{data/exhaustive/mammography_d3_box.csv};
        \addplot table [x=num_dropout, y=avg_time, col sep=comma]{data/exhaustive/mammography_d3_disjuncts.csv};

        \nextgroupplot[
            ymode=log,
            ymin=.00001,
            ymax=10,
            ytick={.00001,.0001,.001,.01,.1,1,10},
        ]
        \addplot table [x=num_dropout, y=avg_time, col sep=comma]{data/exhaustive/mammography_d4_box.csv};
        \addplot table [x=num_dropout, y=avg_time, col sep=comma]{data/exhaustive/mammography_d4_disjuncts.csv};

        \nextgroupplot[
            ylabel=\parbox{1in}{\centering Average Max Memory (MB)},
            ymode=log,
            ymax=1000,
        ]
        \addplot table [x=num_dropout, y=avg_max_memory, col sep=comma]{data/exhaustive/mammography_d1_box.csv};
        \addplot table [x=num_dropout, y=avg_max_memory, col sep=comma]{data/exhaustive/mammography_d1_disjuncts.csv};

        \nextgroupplot[
            ymode=log,
            ymax=1000,
        ]
        \addplot table [x=num_dropout, y=avg_max_memory, col sep=comma]{data/exhaustive/mammography_d2_box.csv};
        \addplot table [x=num_dropout, y=avg_max_memory, col sep=comma]{data/exhaustive/mammography_d2_disjuncts.csv};

        \nextgroupplot[
            ymode=log,
            ymax=1000,
        ]
        \addplot table [x=num_dropout, y=avg_max_memory, col sep=comma]{data/exhaustive/mammography_d3_box.csv};
        \addplot table [x=num_dropout, y=avg_max_memory, col sep=comma]{data/exhaustive/mammography_d3_disjuncts.csv};

        \nextgroupplot[
            ymode=log,
            ymax=1000,
        ]
        \addplot table [x=num_dropout, y=avg_max_memory, col sep=comma]{data/exhaustive/mammography_d4_box.csv};
        \addplot table [x=num_dropout, y=avg_max_memory, col sep=comma]{data/exhaustive/mammography_d4_disjuncts.csv};
    \end{groupplot}
\end{tikzpicture}
\caption{\mammography \label{fig:mammography-details}}
\end{figure*}

%% file: data/full_wdbc.tex
\begin{figure*}
\centering
\small
\pgfplotsset{filter discard warning=false}
\begin{tikzpicture}
    \begin{groupplot}[
            group style={
                group size=4 by 3,
                horizontal sep=.18in,
                vertical sep=.18in,
                ylabels at=edge left,
                yticklabels at=edge left,
                xlabels at=edge bottom,
                xticklabels at=edge bottom
            },
            width=2in,
            height=1.5in,
            xlabel near ticks,
            ylabel near ticks,
            xlabel=Poisoning $n$,
            xmode=log,
            xtick={1,2,4,8,16,32,64,128,256,512,1024,2048},
            log ticks with fixed point,
            xmin=1,
            xmax=64,
        ]

        \nextgroupplot[
            title=Depth 1,
            ylabel=\# Verified,
            ymin=0,
            ymax=113,
        ]
        \addplot table [x=num_dropout, y=num_verified, col sep=comma]{data/exhaustive/wdbc_d1_box.csv};
        \addplot table [x=num_dropout, y=num_verified, col sep=comma]{data/exhaustive/wdbc_d1_disjuncts.csv};

        \nextgroupplot[
            title=Depth 2,
            ymin=0,
            ymax=113,
        ]
        \addplot table [x=num_dropout, y=num_verified, col sep=comma]{data/exhaustive/wdbc_d2_box.csv};
        \addplot table [x=num_dropout, y=num_verified, col sep=comma]{data/exhaustive/wdbc_d2_disjuncts.csv};

        \nextgroupplot[
            title=Depth 3,
            ymin=0,
            ymax=113,
        ]
        \addplot table [x=num_dropout, y=num_verified, col sep=comma]{data/exhaustive/wdbc_d3_box.csv};
        \addplot table [x=num_dropout, y=num_verified, col sep=comma]{data/exhaustive/wdbc_d3_disjuncts.csv};

        \nextgroupplot[
            title=Depth 4,
            ymin=0,
            ymax=113,
        ]
        \addplot table [x=num_dropout, y=num_verified, col sep=comma]{data/exhaustive/wdbc_d4_box.csv};
        \addplot table [x=num_dropout, y=num_verified, col sep=comma]{data/exhaustive/wdbc_d4_disjuncts.csv};
        \legend{Box, Disjuncts};

        \nextgroupplot[
            ylabel=Average Time (s),
            ymode=log,
            ymin=.1,
            ymax=200,
        ]
        \addplot table [x=num_dropout, y=avg_time, col sep=comma]{data/exhaustive/wdbc_d1_box.csv};
        \addplot table [x=num_dropout, y=avg_time, col sep=comma]{data/exhaustive/wdbc_d1_disjuncts.csv};

        \nextgroupplot[
            ymode=log,
            ymin=.1,
            ymax=200,
        ]
        \addplot table [x=num_dropout, y=avg_time, col sep=comma]{data/exhaustive/wdbc_d2_box.csv};
        \addplot table [x=num_dropout, y=avg_time, col sep=comma]{data/exhaustive/wdbc_d2_disjuncts.csv};

        \nextgroupplot[
            ymode=log,
            ymin=.1,
            ymax=200,
        ]
        \addplot table [x=num_dropout, y=avg_time, col sep=comma]{data/exhaustive/wdbc_d3_box.csv};
        \addplot table [x=num_dropout, y=avg_time, col sep=comma]{data/exhaustive/wdbc_d3_disjuncts.csv};

        \nextgroupplot[
            ymode=log,
            ymin=.1,
            ymax=200,
        ]
        \addplot table [x=num_dropout, y=avg_time, col sep=comma]{data/exhaustive/wdbc_d4_box.csv};
        \addplot table [x=num_dropout, y=avg_time, col sep=comma]{data/exhaustive/wdbc_d4_disjuncts.csv};

        \nextgroupplot[
            ylabel=\parbox{1in}{\centering Average Max Memory (MB)},
            ymode=log,
            ymax=100000,
            ytick={1,10,100,1000,10000,100000},
        ]
        \addplot table [x=num_dropout, y=avg_max_memory, col sep=comma]{data/exhaustive/wdbc_d1_box.csv};
        \addplot table [x=num_dropout, y=avg_max_memory, col sep=comma]{data/exhaustive/wdbc_d1_disjuncts.csv};

        \nextgroupplot[
            ymode=log,
            ymax=100000,
            ytick={1,10,100,1000,10000,100000},
        ]
        \addplot table [x=num_dropout, y=avg_max_memory, col sep=comma]{data/exhaustive/wdbc_d2_box.csv};
        \addplot table [x=num_dropout, y=avg_max_memory, col sep=comma]{data/exhaustive/wdbc_d2_disjuncts.csv};

        \nextgroupplot[
            ymode=log,
            ymax=100000,
            ytick={1,10,100,1000,10000,100000},
        ]
        \addplot table [x=num_dropout, y=avg_max_memory, col sep=comma]{data/exhaustive/wdbc_d3_box.csv};
        \addplot table [x=num_dropout, y=avg_max_memory, col sep=comma]{data/exhaustive/wdbc_d3_disjuncts.csv};

        \nextgroupplot[
            ymode=log,
            ymax=100000,
            ytick={1,10,100,1000,10000,100000},
        ]
        \addplot table [x=num_dropout, y=avg_max_memory, col sep=comma]{data/exhaustive/wdbc_d4_box.csv};
        \addplot table [x=num_dropout, y=avg_max_memory, col sep=comma]{data/exhaustive/wdbc_d4_disjuncts.csv};
    \end{groupplot}
\end{tikzpicture}
\caption{\wdbc \label{fig:wdbc-details}}
\end{figure*}

%% file: data/full_mnist_real.tex
\begin{figure*}
\centering
\small
\pgfplotsset{filter discard warning=false}
\begin{tikzpicture}
    \begin{groupplot}[
            group style={
                group size=4 by 3,
                horizontal sep=.18in,
                vertical sep=.18in,
                ylabels at=edge left,
                yticklabels at=edge left,
                xlabels at=edge bottom,
                xticklabels at=edge bottom
            },
            width=2in,
            height=1.5in,
            xlabel near ticks,
            ylabel near ticks,
            xlabel=Poisoning $n$,
            xmode=log,
            xtick={1,8,64,512},
            minor xtick={1,2,4,8,16,32,64,128,256,512,1024,2048,4096},
            log ticks with fixed point,
            xmin=1,
            xmax=2048,
        ]

        \nextgroupplot[
            title=Depth 1,
            ylabel=\# Verified,
            ymin=0,
            ymax=100,
        ]
        \addplot table [x=num_dropout, y=num_verified, col sep=comma]{data/exhaustive/mnist_1_7_d1_box.csv};
        \addplot table [x=num_dropout, y=num_verified, col sep=comma]{data/exhaustive/mnist_1_7_d1_disjuncts.csv};

        \nextgroupplot[
            title=Depth 2,
            ymin=0,
            ymax=100,
        ]
        \addplot table [x=num_dropout, y=num_verified, col sep=comma]{data/exhaustive/mnist_1_7_d2_box.csv};
        \addplot table [x=num_dropout, y=num_verified, col sep=comma]{data/exhaustive/mnist_1_7_d2_disjuncts.csv};

        \nextgroupplot[
            title=Depth 3,
            ymin=0,
            ymax=100,
        ]
        \addplot table [x=num_dropout, y=num_verified, col sep=comma]{data/exhaustive/mnist_1_7_d3_box.csv};
        \addplot table [x=num_dropout, y=num_verified, col sep=comma]{data/exhaustive/mnist_1_7_d3_disjuncts.csv};

        \nextgroupplot[
            title=Depth 4,
            ymin=0,
            ymax=100,
        ]
        \addplot table [x=num_dropout, y=num_verified, col sep=comma]{data/exhaustive/mnist_1_7_d4_box.csv};
        \addplot table [x=num_dropout, y=num_verified, col sep=comma]{data/exhaustive/mnist_1_7_d4_disjuncts.csv};
        \legend{Box, Disjuncts};

        \nextgroupplot[
            ylabel=Average Time (s),
            ymode=log,
            ymin=1,
            ymax=10000,
            ytick={1,10,100,1000,10000},
        ]
        \addplot table [x=num_dropout, y=avg_time, col sep=comma]{data/exhaustive/mnist_1_7_d1_box.csv};
        \addplot table [x=num_dropout, y=avg_time, col sep=comma]{data/exhaustive/mnist_1_7_d1_disjuncts.csv};
        \addplot[dotted, domain=1:4096] {3600};
        \legend{ , , Timeout};

        \nextgroupplot[
            ymode=log,
            ymin=1,
            ymax=10000,
            ytick={1,10,100,1000,10000},
        ]
        \addplot table [x=num_dropout, y=avg_time, col sep=comma]{data/exhaustive/mnist_1_7_d2_box.csv};
        \addplot table [x=num_dropout, y=avg_time, col sep=comma]{data/exhaustive/mnist_1_7_d2_disjuncts.csv};
        \addplot[dotted, domain=1:4096] {3600};

        \nextgroupplot[
            ymode=log,
            ymin=1,
            ymax=10000,
            ytick={1,10,100,1000,10000},
        ]
        \addplot table [x=num_dropout, y=avg_time, col sep=comma]{data/exhaustive/mnist_1_7_d3_box.csv};
        \addplot table [x=num_dropout, y=avg_time, col sep=comma]{data/exhaustive/mnist_1_7_d3_disjuncts.csv};
        \addplot[dotted, domain=1:4096] {3600};

        \nextgroupplot[
            ymode=log,
            ymin=1,
            ymax=10000,
            ytick={1,10,100,1000,10000},
        ]
        \addplot table [x=num_dropout, y=avg_time, col sep=comma]{data/exhaustive/mnist_1_7_d4_box.csv};
        \addplot table [x=num_dropout, y=avg_time, col sep=comma]{data/exhaustive/mnist_1_7_d4_disjuncts.csv};
        \addplot[dotted, domain=1:4096] {3600};

        \nextgroupplot[
            ylabel=\parbox{1in}{\centering Average Max Memory (MB)},
            ymode=log,
            ymax=200000,
        ]
        \addplot table [x=num_dropout, y=avg_max_memory, col sep=comma]{data/exhaustive/mnist_1_7_d1_box.csv};
        \addplot table [x=num_dropout, y=avg_max_memory, col sep=comma]{data/exhaustive/mnist_1_7_d1_disjuncts.csv};
        \addplot[dotted, domain=1:4096] {140000};
        \legend{ , , OOM};

        \nextgroupplot[
            ymode=log,
            ymax=200000,
        ]
        \addplot table [x=num_dropout, y=avg_max_memory, col sep=comma]{data/exhaustive/mnist_1_7_d2_box.csv};
        \addplot table [x=num_dropout, y=avg_max_memory, col sep=comma]{data/exhaustive/mnist_1_7_d2_disjuncts.csv};
        \addplot[dotted, domain=1:4096] {140000};

        \nextgroupplot[
            ymode=log,
            ymax=200000,
        ]
        \addplot table [x=num_dropout, y=avg_max_memory, col sep=comma]{data/exhaustive/mnist_1_7_d3_box.csv};
        \addplot table [x=num_dropout, y=avg_max_memory, col sep=comma]{data/exhaustive/mnist_1_7_d3_disjuncts.csv};
        \addplot[dotted, domain=1:4096] {140000};

        \nextgroupplot[
            ymode=log,
            ymax=200000,
        ]
        \addplot table [x=num_dropout, y=avg_max_memory, col sep=comma]{data/exhaustive/mnist_1_7_d4_box.csv};
        \addplot table [x=num_dropout, y=avg_max_memory, col sep=comma]{data/exhaustive/mnist_1_7_d4_disjuncts.csv};
        \addplot[dotted, domain=1:4096] {140000};
    \end{groupplot}
\end{tikzpicture}
\caption{\mnistreal \label{fig:mnistreal-details}}
\end{figure*}